\documentclass{article}[a4paper]
\usepackage{enumerate}
\usepackage{amssymb}
\usepackage{amsmath}
\usepackage{mathrsfs}
\usepackage{graphicx}
\usepackage{subfigure}
\usepackage{color}
\usepackage{cleveref}
\usepackage{amsthm}
\usepackage{authblk}
\usepackage{fullpage}

\newtheorem{theorem}{Theorem}

\newtheorem{proposition}[theorem]{Proposition}
\theoremstyle{remark}
\newtheorem{remark}[theorem]{Remark}
\theoremstyle{definition}

\usepackage{forest}
\usepackage{algorithm}
\usepackage{algorithmic}
\usepackage{mathtools}
\newcommand{\E}{\mathbb{E}}
\newcommand{\ind}[1]{\mathsf{1}{\left\{#1\right\}}}
\newcommand{\cF}{\mathcal{F}}

\newcommand{\hI}{\hat{I}}

\newcommand{\bI}{\overline{I}}
\newcommand{\ba}{\overline{\alpha}}
\newcommand{\yS}{y_*}
\newcommand{\zS}{z_*}
\newcommand{\vS}{v_*}
\newcommand{\Real}{\mathbb{R}}

\title{
Online Algorithms for Estimating Change Rates of Web Pages\footnote{A shorter version \cite{avrachenkov2020change} of this paper appeared in the proceedings of VALUETOOLS 2020 conference. The novel contributions here include i.) an additional change rate estimation scheme  (this is a stochastic approximation scheme with momentum)
and its analysis and ii.) additional experiments including one that compares the performance of all our estimators based on real data (Wikitraces).} \footnote{This is the author version of the paper accepted to the {\it International Journal of Performance Evaluation}, Elsevier.}
}

\author[1]{Konstantin Avrachenkov}
\author[1]{Kishor Patil}
\author[2]{Gugan Thoppe}
\affil[1]{INRIA Sophia Antipolis, France 06902} 
\affil[2]{Indian Institute of Science, Bengaluru, India 560012}
\affil[ ]{\textit {k.avrachenkov@inria.fr, kishor88k@gmail.com,  gthoppe@iisc.ac.in}}

\begin{document}
\maketitle

\begin{abstract}
A search engine maintains local copies of different web pages to provide quick search results. This local cache is kept up-to-date by a web crawler that frequently visits these different pages to track changes in them.  Ideally, the local copy should be updated as soon as a page changes on the web. However, finite bandwidth availability and server restrictions limit how frequently different pages can be crawled. This brings forth the following optimization problem: maximize the freshness of the local cache subject to the crawling frequencies being within prescribed bounds. While tractable algorithms do exist to solve this problem, these either assume the knowledge of exact page change rates or use inefficient methods such as MLE for estimating the same. We address this issue here. 

We provide three novel schemes for online estimation of page change rates, all of which have extremely low running times per iteration. The first is based on the law of large numbers and the second on stochastic approximation. The third is an extension of the second and includes a heavy-ball momentum term. 
%
All these schemes only need partial information about the page change process, i.e., they only need to know if the page has changed or not since the last crawled  instance. Our main theoretical results concern asymptotic convergence and convergence rates of these three schemes. In fact, our work is the first to show convergence of the original stochastic heavy-ball method when neither the gradient nor the  noise variance is uniformly bounded. We also provide some numerical experiments (based on real and synthetic data) to demonstrate the superiority of our proposed estimators over existing ones such as MLE. We emphasize that our algorithms are also readily applicable to the synchronization of databases and network inventory management.
\end{abstract}

\section{Introduction}
The worldwide web is a highly complex entity: it has a lot of interlinked information, and both the information and the links keep evolving. Nevertheless, even in this challenging setup, one still expects a search engine to provide accurate and up-to-date search results instantaneously. To fulfill this expectation, a search engine maintains a local cache of important web pages, that it updates frequently by using a crawler (also referred to as a web spider or a web robot). Specifically, the job of a crawler \cite{heydon1999mercator, castillo2005effective, Olston2010, Edwards2001, avrachenkov2011optimal} is (a) to access various pages on the web at specific frequencies so as to determine if any changes have happened to the content since the last crawled instance; and (b) to update the local cache whenever there is a change\footnote{A web crawler is also supposed to discover new pages, but we don't focus on this task in this work.}. There are, however, two key constraints on the different crawling frequencies. The first is due to limitations on the available bandwidth. The second one, known as the politeness constraint, arises because of the bounds placed by servers on the number of pages that can be accessed in a short amount of time. The search engine thus needs to solve the following optimization problem: maximize the freshness of the local database subject to the crawling frequencies satisfying the above constraints. 

In the early 2000s, the web crawling problem used to be  formulated as follows \cite{cho2000evolution, cho2000synchronizing, matloff2005}. The whole web consists of $n$ pages, all have equal importance, and there are no politeness constraints. Further, the times at which different pages change are independent Poisson point processes with different rates \cite{cho2000evolution, brewington2000dynamic}. On the search engine side, the local cache consists of a copy of each of these $n$ pages. Each copy is updated at regular intervals of time by  crawling the original page at a certain (known) frequency. Finally, the  freshness of the local cache at time $t \geq 0$ is defined to be $r \in [0, 1]$ if $r$ fraction of the local elements matches the actual versions on the web. The goal then is to find an update policy that maximizes the time-averaged freshness of the local cache and, also, satisfies the bandwidth constraint. 

Finding an exact solution to this problem is hard. Hence, numerical solutions were obtained in \cite{cho2000evolution, cho2000synchronizing} for small values of $n.$ These showed that the optimal crawling policy could be very different to both the uniform as well as the proportional policy, i.e., crawling each page at the same frequency  or at one that is proportional to its change rate. In fact, somewhat surprisingly, it was also found that the optimal policy may often include ignoring pages that change too frequently, i.e., not crawling them at all.

In 2003, the freshness definition was modified to include different weights for different pages depending on their importance, e.g., represented as the frequency of requests for different pages \cite{cho2003effective}. This was done in line with the view that only a finite number of pages can be crawled in any given time frame; hence, to improve the utility of the local database, the freshness criteria should be biased more towards important pages. Numerical solutions, again for small $n,$ confirm that page weights do substantially influence the optimal crawling policy. 

While the general $n$ case is still unsolved, a recent breakthrough work \cite{Azar2018} showed how an optimal randomized crawling policy can, nonetheless, be found very efficiently (in just $O(n \log n)$ operations). In particular, this solution pertains to the case where, for each web page, even the set of access times forms a Poisson point process. An approach to derandomize this policy to handle the original setup with periodic crawling is also discussed there. In synthetic experiments, this resultant policy is claimed to show performances very similar to the one obtained via numerical solutions. This work was recently extended to cover the case with politeness constraints as well \cite{kolobov2019optimal}. 

There is also a separate study  \cite{avrachenkov2016whittle,nino2014dynamic} which provides a Whittle index based dynamic programming approach to optimize the schedule of a web crawler. In that approach, the page/catalog freshness estimate also influences the optimal crawling policy. 

As can be seen, several algorithms do exist to determine the optimal crawling policy. However, they either presume prior knowledge of the exact page change rates, which is unrealistic in practice, or, alternatively, use inefficient ideas for estimating the same. We now provide a brief overview of such approaches and the issues that plague them. 

To the best of our knowledge, three other estimators exist in the literature: the naive estimator \cite{douglis1997rate, wills1999towards, wolman1999scale}, the Maximum Likelihood Estimator (MLE) \cite{cho2003estimating}, and the Moment Matching (MM) estimator \cite{upadhyay2019}. The naive estimator is simply the ratio of the observed number of changes to the total monitoring time period. This is clearly biased since the crawler only has access to partial information about the page change process (remember, it only gets to see if a page has changed or not since the last crawled instance). To overcome this bias issue, MLE instead estimates the rate of change by identifying the parameter value that maximizes the likelihood of the page change observations. This idea performs quite well in experiments; in fact, it also works when access to a page is only possible at irregular intervals of time. However, MLE lacks a closed form expression and suffers from two issues: (a) instability, i.e., the estimator value equals $\infty$ as long as a page change is detected in every access; and (b) computational intractability, i.e., the estimate needs to be recomputed from scratch each time a new observation is made. The latter makes MLE impractical to use when the data set of observations is quite large. Finally, in the MM estimator, one looks at the fraction of times no changes were detected during page accesses and then, using a moment matching method, estimates the change rate. Unfortunately, like MLE, the MM estimator also suffers from instability and computational issues.

Our exact problem statement and the main contributions can now be summarized as follows. We consider a single page and, as in \cite{Azar2018}, presume that the page change times and page access times are independent homogeneous Poisson point processes. Thus, each of these processes can be characterized by a single parameter, which we denote here by $\Delta$ and $p,$ respectively. Importantly, we assume that only $p$ is known. We then develop three approaches for online estimation of $\Delta,$ which only need to know if this page has changed or not between two successive accesses. The key word here is `online'. This means, unlike MLE and the MM estimator, our estimates can be incrementally updated using extremely simple, low cost formulas as and when a new observation becomes available. Thus, our estimators do not face computational issues of the kind mentioned above. Also, they do not face any instability issues.

Our first estimator uses the Law of Large Numbers (LLN), while the second and third estimators are based on Stochastic Approximation (SA) principles. Specifically, the update rule for the first estimator is derived using a formula for the probability that there is a page change between two successive accesses. In contrast, the second estimator is constructed via a standard trick in SA. A key ingredient there is a function that is carefully chosen so that it satisfies two properties: (a) noisy estimates of its value for any given input can be easily obtained; and (b) its expected value is linear and, importantly, $\Delta$ is its unique zero. The update rule for the third estimator is similar to that of the second one, except that it has an additional momentum term (in the heavy-ball sense). As we show in Section~\ref{s:SAM.Literature}, it is also possible to view our second and third estimators as a Stochastic Gradient Descent (SGD) method and as an SGD method with heavy-ball momentum, respectively. We emphasize that even though we present
our results in the context of web crawling, our algorithms are equally applicable to the synchronization of databases \cite{cho2000synchronizing} and the problem of network inventory management \cite{andrews2020tracking}.


Our main theoretical result is that all our estimators almost surely (a.s.) converge to $\Delta;$ thus, they all are asymptotically consistent. As far as we know, our result concerning the third estimator is the first to show convergence of an SGD method with heavy-ball momentum when neither the gradient nor the  noise variance is a priori assumed to be uniformly bounded. While similar settings have also been dealt with in \cite{gadat2018stochastic}, the analysis there concerns the stochastic analogue of a modified heavy-ball method and not the original one that was proposed in  \cite{polyak1964some}.  Separately, we also derive the convergence rates of the first two estimators in the expected error sense. Based on the existing literature, we also provide a loose guess on the convergence rate of the third estimator. 

We also provide numerical simulations to compare the performance of our online schemes to each other and also to that of the (offline) MLE estimator. From these experiments, it can be explicitly seen that our estimators give performances comparable to that of MLE. This was a bit surprise to us since our estimators, compared to MLE, have extremely low running times per iteration. Also, unlike MLE, they ignore the actual lengths of intervals between two page accesses. Among our three estimators, LLN and SAM show similar performances and both typically outperform our SA estimator. In particular, the  momentum in the third estimator helps in accelerating the estimation whenever $p \ll \Delta$ (the rate at which the page is accessed is much smaller than the rate at which it changes). Our experiments are based on both real (Wikipedia traces) as well as synthetic data sets. In the experiment using Wikipedia traces, we also verify our modeling assumption that the page change process is a Poisson point process. 

The rest of this paper is organized as follows. The next section provides a formal summary of this work in terms of the setup, goals, and key contributions. It also gives explicit update rules for all of our online schemes. In Section \ref{sec3}, we formally analyse  their convergence and the rates of convergence. The numerical experiments discussed above are given in Section~\ref{sec4:numericalExpts}. Then, in Section~\ref{sec:Motivation}, we provide some motivation on how one can use our estimates to find the optimal crawling rates. Finally, we conclude in Section~\ref{sec7} with some future directions.

\section{Setup, Goal, and Key Contributions}
\label{sec2}
The three topics are individually described below. \\[-1ex]

\noindent \textbf{Setup}: Without loss of generality, we work with a single web page. We presume that the actual times at which this page changes is a time-homogeneous Poisson point process in $[0, \infty)$ with a constant but unknown rate $\Delta.$ Independently of everything else, this page is crawled (accessed) at the random instances $\{t_{k}\}_{k\geq 0} \subset [0, \infty),$ where $t_{0} = 0$ and the inter-arrival times, i.e., $\{t_{k} - t_{k - 1}\}_{k \geq 1},$ are IID exponential random variables with a known rate $p.$ Thus, the times at which this page is crawled is also a time-homogeneous Poisson point process but with rate $p.$ At time instance $t_{k},$ we get to know if the page got modified or not in the interval $(t_{k - 1}, t_{k}],$ i.e., we can access the value of the indicator
\begin{equation*}
I_{k} := \begin{cases*}
                1, & if the page got modified in $(t_{k-1},t_{k}],$ \\
                0, & otherwise.
            \end{cases*}
\end{equation*}

The above assumptions are standard in the crawling literature. Nevertheless, we now provide a short justification for the same. Our assumption that the page change process is a Poisson point process is based on the experimental evidence collected in \cite{brewington2000dynamic, brewington2000keeping, cho2000evolution}. An additional validation is provided by us in this work. Specifically, we selected an arbitrary page from the list of frequently edited Wikipedia pages. We extracted the complete history of this web page (exact dates and times of different changes) for a period of five months (April 01, 2020 to August 31, 2020). Thereafter, we calculated the time between successive changes and then used this data to produce a Q-Q plot. This plot confirms that the set of quantiles for the actual data indeed matches linearly with the quantiles of exponential distribution, as predicted. Further details about this experiment can be found in Section~\ref{sec4:numericalExpts}. Some generalized models for the page change process  have also been considered in the literature \cite{matloff2005, singh2007}; however, we do not pursue them here. 

Our assumption on $\{I_k\}$ is based on the fact that a crawler can only access incomplete knowledge about the page change process. In particular, a crawler does not know when and how many times a page has changed between two crawling instances. Instead, all it can track is the status of a page at each crawling instance and know if it has changed or not with respect to the previous access.  Sometimes, it is possible to also know the time at which the page was last modified \cite{castillo2005effective, cho2003estimating}, but we do not consider this case here. \\[-1ex]

\noindent \textbf{Goal}: Develop online algorithms for estimating $\Delta$ in the above setup. The motivation for doing this is that such estimates can then be used to estimate the optimal crawling rates \cite{Azar2018,kolobov2019staying}; see Section~\ref{sec:Motivation} for more details on this. \\[-1ex]


\noindent \textbf{Key Contributions}:
We provide three online methods for estimating the page change rate $\Delta.$ The first is based on the law of large numbers, while the second and third are based on stochastic approximation theory, with the third one having an additional momentum component. If $\{x_k\},$ $\{y_k\},$ and $\{z_k\},$ denote the iterates of these three methods, respectively, then their update rules are as shown below. 
\begin{itemize}
    \item \emph{LLN Estimator}: Its $k$-th estimate is given by
    \begin{equation}
        \label{eqn:LLN_Est}
        x_{k} = p \hI_k/(k + \alpha_k - \hI_k), \quad k \geq 1.
    \end{equation}
    Here, $\hI_k = \sum_{j = 1}^k I_j;$ hence, $\hI_k = \hI_{k - 1} + I_k.$ Further, $\{\alpha_k\}$ is any positive sequence satisfying the conditions in Theorem~\ref{thm:LLN_Est}; e.g.,  $\alpha_k$ could be $\log k,$ $\sqrt{k},$ or  identically $1.$  \\[-1ex]
    
    \item \emph{SA Estimator}: Given some initial value $y_0,$ the update rule for the SA estimator is
    \begin{equation}
        \label{eqn:SA_Est}
        y_{k + 1} = y_k + \eta_k[I_{k + 1}(y_k + p) - y_k], \quad k \geq 0.
    \end{equation}
    Here, $\{\eta_k\}$ is any stepsize sequence that satisfies the conditions in Theorem~\ref{thm:SA_Est}. For example, $\eta_k$ could be $1/(k + 1)^\eta$ for some constant $\eta \in (0, 1].$ 
    \item \emph{SAM Estimator $($SA Estimator with Momentum$)$}: Given some initial values $z_0, z_{-1},$ the SAM estimator satisfies
    \begin{equation}
        \label{eqn:SAM_Est}
        z_{k + 1} = z_k + \eta_k[I_{k + 1}(z_k + p) - z_k] + \zeta_k(z_{k} - z_{k - 1}), \quad k \geq 0.
    \end{equation}
    Here, $\{\eta_k\}$ and $\{\zeta_k\}$ are any stepsize sequences that satisfy the conditions given in Theorem~\ref{thm:SAM_Est}. For example, one could pick a $\beta \in (1/2, 1]$ and let $\beta_k = 1/(k + 1)^\beta.$ Then, $\{\eta_k\}$ and $\{\zeta_k\}$ could be $\{1/(k + 1)^\eta\}$ and $\{(\beta_k - \omega \eta_k)/(\beta_{k - 1})\},$ respectively, where $\omega > 0$ is some constant and $\beta + 1/2 < \eta \leq 2\beta.$ While we do not show it, we conjecture that one can also pick $\beta \in (0, 1/2]$ and then choose $\eta$ so that $\beta < \eta \leq 2 \beta.$ Finally, note that if $\beta = \eta$ and $\omega = 1,$ then the asymptotic behaviour of \eqref{eqn:SAM_Est} will resemble that of  \eqref{eqn:SA_Est}; this is because $\zeta_k \equiv 0$ then.
\end{itemize}

We call these methods online because the estimates can be updated on the fly as and when a new observation $I_k$ becomes available. This contrasts the MLE estimator in which one needs to start the calculation from scratch each time a new data point  arrives. Also, unlike MLE, our estimators are never unstable; see Section~\ref{subsec:Comp} for the details. 

Our main results include the following. We show that all our three estimators, i.e., $x_k, y_k,$ and $z_k,$ converge to $\Delta$ a.s. Further, we show that 
\begin{enumerate}
    \item $\E|x_k - \Delta| = O\left(\max\left\{k^{-1/2}, \alpha_k/k\right\}\right),$ and 
    
    \item $\E|y_k - \Delta| = O(k^{-\eta/2})$ if $\eta_k = (k + 1)^\eta$ with $\eta \in (0, 1).$
\end{enumerate}
Separately, based on existing literature \cite{dalal2018finite2TS, dalal2019tale, kaledin2020finite}, we conjecture that $\E|z_k - \Delta| = \tilde{O}(k^{-\beta/2}),$ where $\tilde{O}$ hides logarithmic terms. We also provide several numerical experiments based on real as well as synthetic data  for judging the strength of our three proposed estimators. 

\section{Analysis of the Proposed Online Estimators} \label{sec3}
Here, we formally discuss the convergence and convergence rates of our three estimators. Thereafter, we compare their behaviors with those that already exist in the literature---the Naive estimator, MLE, and the MM estimator. We end with a summary of existing results on stochastic momentum methods and a discussion on how our convergence result for the SAM estimator extends our current understanding of such methods.


\subsection{LLN Estimator}
\label{subsecLLN_Est}
Our first aim here is to obtain a formula for $\E[I_1].$ We shall use this later to motivate the form of our LLN estimator.

Let $\tau_1 = t_1 - t_0 = t_1,$ where the second equality holds  since $t_0 = 0.$ Then, as per our assumptions in Section~\ref{sec2}, $\tau_1$ is an exponential random variable with rate $p.$ Also, $\E[I_1 | \tau_1 = \tau ] = 1- \exp{(-\Delta \tau)}.$ Hence, 
\begin{equation}
\label{eqn:Exp_Ind}
   \E\big[I_1\big] = \Delta/(\Delta + p).
\end{equation}
This gives the desired formula for $\E[I_1].$ 

From this latter calculation, we have
\begin{equation}
\label{e:Actual_Delta_i_Formula}
\Delta = p \E[I_1]/(1 -\E[I_1]).
\end{equation}
Separately, because $\{I_k\}$ is an IID sequence and $\E|I_1| \leq 1$ , it follows from the strong law of large numbers that $\E\big[I_1\big] = \lim_{k \to \infty}  \sum_{j = 1}^k I_j /k \; \text{a.s.}$
Thus,
\[
\Delta = p \frac{\lim_{k \to \infty}\sum_{j = 1}^k I_j/k}{1 - \lim_{k \to \infty}\sum_{j = 1}^k I_j/k} \quad  \text{a.s.}
\]
Consequently, a natural estimator for  $\Delta$ is
\begin{equation}
    x_k' = p \frac{\sum_{j = 1}^k I_j/k}{1 - \sum_{j = 1}^kI_j/k} = p \frac{\hI_k}{k - \hI_k},
\end{equation}
where $\hI_k$ is as defined below \eqref{eqn:LLN_Est}.

Unfortunately, the above estimator faces an instability issue, i.e., $x'_k = \infty$ when $I_1, \ldots, I_k$ are all $1.$ To fix this, one can add a non-zero term in the denominator. The different choices then gives rise to the LLN estimator defined in \eqref{eqn:LLN_Est}.

The following result discusses the convergence and convergence rate of this estimator. 

\begin{theorem}
\label{thm:LLN_Est}
Consider the estimator given in \eqref{eqn:LLN_Est} for some positive sequence $\{\alpha_k\}.$
\begin{enumerate}
    \item If \; $\lim_{k \to \infty} \alpha_k/k= 0,$ then  $\lim_{k \to \infty} x_k = \Delta \; \text{a.s.}$
    
    \item Additionally, if \; $\lim_{k \to \infty} \log(k/\alpha_k)/k = 0,$ then 
    \[
    \E |x_k - \Delta| = O\left(\max\left\{k^{-1/2}, \alpha_k/k\right\}\right).
    \]
\end{enumerate}
\end{theorem}
\begin{proof}
Let $\mu = \E[I_1],$ $\bI_k = \hI_k/k,$ and $\ba_k = \alpha_k/k.$ Then, observe that \eqref{eqn:LLN_Est} can be rewritten as $x_k = p \bI_k/(1 + \ba_k - \bI_k).$
Now, $\lim_{k \to \infty} \bI_k = \mu$ a.s. and $\lim_{k \to \infty} \ba_k = 0;$ the first claim holds due to the strong law of large numbers, while the second one is true due to our assumption. Statement 1. is now easy to see. 

We now derive Statement 2.  From \eqref{e:Actual_Delta_i_Formula}, we have
\[
|x_k - \Delta| =  \left|x_k - p \frac{\mu}{1 - \mu} \right| \leq p\left(A_k + B_k\right),
\]
where 
\[
A_k = \left|\frac{\bI_k}{\ba_k + 1 - \bI_k} -  \frac{\mu}{\ba_k + 1 - \mu}\right| \quad \text{ and }  \quad
B_k = \left|\frac{\mu}{\ba_k + 1 - \mu} - \frac{\mu}{1 - \mu} \right|.
\]
Since $\alpha_k > 0$ and, hence, $\ba_k > 0,$ it follows that
\[
B_k = \ba_k \frac{\mu}{(1 - \mu)(\ba_k + (1 - \mu))} \leq \ba_k \frac{\mu}{(1 - \mu)^2}.
\]
Similarly, 
\[
A_k \leq \left(\frac{1 + \ba_k}{1 - \mu}\right) \left(\frac{|\bI_k - \mu|}{\ba_k + 1 - \bI_k}\right).
\]
It is now easy to see that $\E [B_k] = O(\ba_k).$  The rest of our arguments concern how fast $\E [A_k]$ decays to $0.$

Let $\{\delta_k\}$ be a deterministic sequence that is both non-negative and decays to $0.$ We will describe how to pick this later. Let $k$ be such that $(1 + \delta_k)\mu < 1.$ Then, 
\[
\E\left[\frac{|\bI_k - \mu|}{\ba_k + 1 - \bI_k}\right] \leq \E [C_k] + \E [D_k],
\]
where 
\[
C_k = \frac{|\bI_k - \mu|}{\ba_k + 1 - \bI_k}\ind{\bI_k - \mu \leq \delta_k \mu},
\]
and 
\[
D_k = \frac{|\bI_k - \mu|}{\ba_k + 1 - \bI_k}\ind{\bI_k - \mu \geq \delta_k \mu}.
\]
On the one hand, 
\[ 
\E [C_k] \leq \frac{\E|\bI_k - \mu|}{\ba_k + 1 - (1 + \delta_k)\mu}  \leq \frac{\sqrt{\text{Var}[I_1]}}{\sqrt{k} (\ba_k + 1 - (1 + \delta_k) \mu)}.
\]
On the other hand, since  $|\bI_k - \mu| \leq 2$ and $1 - \bI_k \geq 0,$ it follows by applying the Chernoff bound that
\[
\E [D_k] \leq \frac{2}{\ba_k} \Pr \{\bI_k \geq (1 + \delta_k)\mu \} \leq \frac{2}{\ba_k} \exp\left(-k\delta_k^2  \mu/3\right).
\]

Now, pick $\{\delta_k\}$ so that $\delta_k^2 = 6\log(1/\, \ba_k)/(k \mu) \vee 0$ for all $k \geq 1.$ Notice that this choice is both non-negative and decays to $0$ due to our assumptions on $\{\alpha_k\};$ thus, this is a valid choice. It is now easy to see that $\E[C_k] = O(1/\sqrt{k})$ and $\E [D_k] =O( \ba_k).$ 

The desired result now follows. 
\end{proof}

\subsection{SA Estimator}
\label{subsec:SA}
Let $I$ denote a random variable with the same distribution as $I_1.$ Also, for $y \in \Real,$ let $H(y, I) = I(y + p) - y.$ Next, define $h: \Real \to \Real$ using $h(y) := \E[H(y, I)].$ Clearly, $h(y) = p(\Delta - y)/(\Delta + p);$ further, $\Delta$ is its unique zero. The theory of stochastic approximation then suggests using the update rule given in \eqref{eqn:SA_Est} for estimating $\Delta.$ For later use, also define
\begin{align}
    M_{k + 1} = {} & [I_{k+1}(y_k + p) - y_k] - h(y_k) \nonumber \\
    = {} & \left[I_{k + 1} - \frac{\Delta}{\Delta + p}\right](y_k + p). \label{d:Mart.Diff.Seq}
\end{align}

We now discuss the convergence and convergence rate of \eqref{eqn:SA_Est}.

\begin{theorem}
\label{thm:SA_Est}
Consider the estimator given in \eqref{eqn:SA_Est} for some positive stepsize sequence $\{\eta_k\}.$
\begin{enumerate}
    \item Suppose that $\sum_{k = 0}^\infty \eta_k = \infty$ and $\sum_{k = 0}^\infty \eta_k^2 < \infty.$ Then, $\lim_{k \to \infty} y_k = \Delta$ a.s.

    \item Suppose that $\eta_k = 1/(k + 1)^\eta$ for some constant $\eta \in (0, 1).$ Then,
    \[
        \E |y_k - \Delta| = O\left(k^{-\eta/2}\right).
    \]
\end{enumerate}
\end{theorem}
\begin{proof}
For $k \geq 0,$ consider the $\sigma-$field $\cF_k :=  \sigma(y_j, I_j, j \leq k).$ Then, from \eqref{eqn:Exp_Ind} and the fact that $\{I_k\}$ is an IID sequence, we get 
\[
\E[I_{k + 1}(y_k + p) - y_k |\cF_k] = \frac{\Delta}{\Delta + p}  (y_k + p) - y_k = h(y_k).
\]
Hence, one can rewrite \eqref{eqn:SA_Est} as 
\begin{equation}
\label{eqn:SAForm}
y_{k + 1} = y_k + \eta_k [h(y_k) + M_{k + 1}],
\end{equation}
where $M_{k + 1}$ is as in \eqref{d:Mart.Diff.Seq}.  

Since $\E[M_{k + 1}|\cF_k] = 0$ for all $k \geq 0,$ $\{M_k\}$ is a martingale difference sequence. Consequently, \eqref{eqn:SAForm} is a classical SA algorithm whose limiting ODE is 
\begin{equation}
\label{eqn:limODE}
    \dot{y}(t) = h(y(t)).
\end{equation}

We now make use of Theorem~\ref{thm:1TS.Conv} given in the Appendix to establish Statement~1. Accordingly, we verify the four conditions listed there. The stepsize Condition i.) directly holds due to our assumptions on $\{\eta_k\}.$ With regards to Condition ii.), recall we have already established above that $\{M_k\}$ is a martingale difference sequence with respect to $\{\cF_k\}$. The square-integrability condition holds since $|M_{k + 1}| \leq |y_k| + p$ which, in turn, implies that $\E[|M_{k + 1}|^2|\cF_k] \leq 2(p^2 \vee 1)(1 + |y_k|^2),$ as desired. Next, due to linearity, $h$ is trivially Lipschitz continuous. Further, $h(y) = 0$ if and only if $y = \Delta.$ This shows that $\Delta$ is the unique equilibrium point of \eqref{eqn:limODE}. Now, because the coefficient of $y$ in $h(y)$ is negative, it also follows that $\Delta$ is the unique globally asymptotically stable equilibrium of \eqref{eqn:limODE}. This verifies Condition iii.). We finally consider Condition iv.) Let $h_\infty(y) := -y p/(\Delta + p).$ Then, clearly, $h_c \to h_\infty$ uniformly on compacts as $c \to \infty.$ Furthermore, since the coefficient of $y$ is negative in the definition of $h_\infty,$ it is easy to see that the origin is the unique globally asymptotically stable equilibrium of the ODE $\dot{y}(t) = h_\infty(y(t)),$ as required. Statement 1. now follows. 

We now sketch a proof for Statement 2. First, note that
\[
y_{k + 1} - \Delta = (1 - a \eta_k) (y_k - \Delta) + \eta_k M_{k + 1},
\]
where $a = p/(\Delta + p).$ Now, since $\E[M_{k + 1}|\cF_k] = 0,$ we have
\[
\E[(y_{k + 1} - \Delta)^2|\cF_k] = (1 - a \eta_k)^2(y_k - \Delta)^2 + \eta_k^2 \E[M_{k + 1}^2|\cF_k].
\]
Recall that $\E[M_{k + 1}^2 |\cF_k] \leq C(1 + y_k^2)$ for some constant $C \geq 0.$ By substituting this above and then repeating all the steps from the proof of  \cite[Theorem~3.1]{dalal2018finite1TS}, it is not difficult to see that Statement 2 holds as well.
\end{proof}

\subsection{SA Estimator with Momentum}
As stated before, our SAM estimator is the SA estimator discussed above with an additional heavy-ball momentum term. Simulations in Section~\ref{sec4:numericalExpts} show that this simple modification results in a drastic improvement in performance.

We now discuss the convergence of the SAM estimator under the assumption that, for $k \geq 0,$
\begin{equation}
    \label{d:zeta_k}
    \zeta_k = \frac{\beta_k - \omega \eta_k}{\beta_{k - 1}},
\end{equation}
where $\omega > 0$ is some constant  and $\{\beta_k\}$ is some positive real sequence.  By substituting \eqref{d:zeta_k} and letting $u_k = (z_k - z_{k - 1})/\beta_{k - 1},$ observe that the update rule in \eqref{eqn:SAM_Est} can be rewritten as 
\[
    u_{k + 1} = u_k + \gamma_k \left [I_{k + 1}(z_k + p_i) - z_k\right] - \omega \gamma_k u_k,
\]
where $\gamma_k := \eta_k/\beta_k.$ 

For $k \geq 0,$ let $M_{k + 1}$ be as in \eqref{d:Mart.Diff.Seq}. Also, let $\cF_k$ denote the $\sigma$-field $\sigma(z_0, u_0, I_1, \ldots, I_k).$ Clearly, $u_k, z_k \in \cF_k$ and  $\E[M_{k + 1} |\cF_k] = 0.$ Hence, $\{M_k\}$ is again a martingale difference sequence with respect to the filtration $\{\cF_k\}.$ Furthermore, since $|M_{k + 1}| \leq |z_k| + p,$ we have 
\begin{equation}
\label{e:1TS.martingale.cond}
    \E[|M_{k + 1}|^2|\cF_k] \leq 2 (p^2 \vee 1)(1 + |z_k|^2).
\end{equation}

As before, let $a = p/(\Delta + p).$ Also, let $b = \Delta p/(\Delta + p)$ and $\epsilon_{k} = u_{k + 1} - u_k$ for $k \geq 0.$ It is then easy to see that one can write down \eqref{eqn:SAM_Est} in terms of the following two update rules:
\begin{align}
    u_{k + 1} = {} & u_k + \gamma_k [h(u_k, z_k) + M_{k + 1}] \label{e:u_Upd}\\
    z_{k + 1} = {} & z_k +  \beta_k [g(u_k, z_k) +  \epsilon_{k}] \label{e:z_Upd},
\end{align}
where $h: \Real^2 \to \Real$ and $g: \Real^2 \to \Real$ are the linear functions given by
\[
    h(u, z) = b - \omega u - az   \quad \text{ and }  \quad g(u, z) = u.
\]

\begin{theorem}
\label{thm:SAM_Est}
Consider the SAM estimator given in \eqref{eqn:SAM_Est} with $\zeta_k$ of the form given in \eqref{d:zeta_k}. Then $z_k \to \Delta$ a.s., if one of the following conditions holds true.

\begin{enumerate}
     \item \emph{One-timescale }: $\sum_{k \geq 0} \beta_k = \infty,$ \; $\sum_{k \geq 0} \beta_k^2 < \infty,$ \; and \; $ \beta_k = \gamma_k.$

    \item \emph{Two-timescale}: $\sum_{k \geq 0} \beta_k = \sum_{k \geq 0} \gamma_k = \infty,$ \; $\sum_{k \geq 0} \left(\beta_k^2 + \gamma_k^2 \right) < \infty,$ \; and \; $\lim_{k \to \infty} \dfrac{\beta_k}{\gamma_k} = 0.$
\end{enumerate}
In both these cases, recall that $\gamma_k = \eta_k/\beta_k.$
\end{theorem}

We state a few remarks concerning this result before discussing its proof.

\begin{remark}
\label{rem:SAM.Stepsizes}
 Examples of $\{\eta_k\}$ and $\{\beta_k\}$ sequences such that the above conditions are satisfied include the following. 
\begin{itemize}
    \item \emph{One-timescale}: $\beta_k = 1/(k + 1)^\beta$ with $\beta \in (1/2, 1]$ and $\eta_k = 1/(k + 1)^\eta$ with $\eta = 2 \beta.$
    
    \item \emph{Two-timescale}: $\beta_k = 1/(k + 1)^\beta$ with $\beta \in (1/2, 1]$ and $\eta_k = 1/(k + 1)^\eta$ with $\frac{1}{2} + \beta < \eta < 2 \beta.$
\end{itemize}
In either case, note that $\lim_{k \to \infty} \zeta_k = 1.$
\end{remark}

\begin{remark}
The justification for the names given above for the two sets of conditions is as follows. 
Under the first set of conditions, the update rules in \eqref{e:u_Upd} and \eqref{e:z_Upd} indeed behave like a one-timescale stochastic approximation algorithm, i.e., both $u_k$ and $z_k$ move on the same timescale. On the other hand, under the second set of conditions, \eqref{e:u_Upd} and \eqref{e:z_Upd}, it behaves like a two-timescale stochastic approximation algorithm. This is because $\beta_k$  decays to $0$ at a much faster rate than  $\gamma_k,$ in turn implying that the changes in $\{z_k\},$ i.e., $\{z_{k + 1} - z_k\}$ are of a smaller magnitude than that in $\{u_k\}.$
\end{remark}

\begin{remark}
In the spirit of the above remark, a natural question to consider is the following. Can one pick $\{\eta_k\}$ and $\{\beta_k\}$ so that $\eta_k/\beta_k^2 \to 0$ or, equivalently, $\gamma_k/\beta_k \to 0?$ That is, can one pick the stepsizes so that $u_k$ now becomes the slowly moving update relative to $z_k?$ The answer to this question seems to be no. This is because a couple of sufficient conditions needed to guarantee convergence (e.g., Condition iii.) and iv.) in Theorem~\ref{thm:2TS.Conv}) would no longer hold true in this new setup. Furthermore, simulations seem to suggest that the iterates, in fact, race to infinity. 
\end{remark}

\begin{remark}
Another question to consider is the following. Can one pick $\omega,$ $\{\beta_k\},$ and $\{\eta_k\}$ so that $\zeta_k \to \zeta,$ where $\zeta$ is a constant in $(0, 1)?$ In particular, can one choose $\omega = (1 - \zeta),$ $\beta_k = 1/(k + 1)^\beta$ with $\beta \in (1/2, 1]$ and then pick $\eta_k = 1/(k + 1)^\beta$ (i.e., $\eta = \beta$) so that $\zeta_k \to \zeta?$ The answer to this second question does not seem to be clear. This is because $\lim_{k \to \infty}\gamma_k$ would then equal $1.$ Consequently, again, one of the sufficient conditions to guarantee convergence (e.g., condition i.) of Theorem~\ref{thm:2TS.Conv}) would no longer hold. However, simulations in this case do show some promise. %
\end{remark}

\begin{remark}
\label{rem:ConvRate.SAM}
Based on the existing literature on convergence rates for one-timescale and two-timescale linear stochastic approximation \cite{dalal2018finite1TS, dalal2018finite2TS,dalal2019tale,kaledin2020finite}, one can conjecture that $\E|z_k - \Delta| = \tilde{O}(k^{-\beta/2})$ when $\{\beta_k\}$ and $\{\eta_k\}$ are chosen as described in Remark~\ref{rem:SAM.Stepsizes}. This implies the optimal convergence rate would then again be $\tilde{O}(1/\sqrt{k}),$ which matches the bound we have obtained in Theorem~\ref{thm:SA_Est} for the SA estimator. However, it is possible that this bound  may not be tight in the case of the SAM estimator. The is because \eqref{e:z_Upd} lacks the martingale difference term and, typically, these are the kind of terms that dictate the convergence rates. Furthermore, simulations in Section~\ref{sec4:numericalExpts} suggest that the SAM estimator always converges much faster than the SA estimator.
\end{remark}   

\begin{proof}[Proof of Theorem~\ref{thm:SAM_Est}]
We discuss the two cases one by one. 

\emph{One-timescale Setup}: In this case, the update rules given in \eqref{e:u_Upd} and \eqref{e:z_Upd}  together form a one-timescale stochastic approximation algorithm. More specifically, if we let $v_k = \begin{bmatrix}
u_k \\ z_k \end{bmatrix},$ then it follows that
\begin{equation}
    \label{eqn:SAM.1TS}
    v_{k + 1} = v_k + \beta_k \left(H(v_k) + \begin{bmatrix} 0 \\ \epsilon_{k} \end{bmatrix} +  \begin{bmatrix} M_{k + 1} \\ 0 \end{bmatrix} \right),
\end{equation}
where $H: \Real^2 \to \Real^2$ is the function defined by
\[
    H(v) = \begin{bmatrix} b \\ 0 \end{bmatrix} - \begin{bmatrix} 
    \omega & a \\ 
    -1 & 0 
    \end{bmatrix} v.
\]

We now verify the four conditions listed in  Theorem~\ref{thm:1TS.Conv} and then make use of Proposition~\ref{prop:1TS.Conv.Perturbed} (both given in the appendix) to show that $v_k \to \begin{bmatrix} 0 \\ \Delta \end{bmatrix} =: \vS$ a.s. This automatically implies $z_k \to \Delta$ a.s., which is what we need to prove. 

Notice that the stepsize in \eqref{eqn:SAM.1TS} is $\beta_k.$ Condition~i.), therefore, trivially holds due to the assumptions made in Statement 1. Next, observe that the martingale difference term in \eqref{eqn:SAM.1TS} is the vector $\begin{bmatrix} M_{k + 1} \\ 0 \end{bmatrix}.$ This, along with  \eqref{e:1TS.martingale.cond} and the statements above it, shows that  Condition ii.) is true as well.

With regards to Condition iii.), first note that $H$ is trivially Lipschitz continuous due to the linearity of both its component functions. Next, since $\Delta = b/a,$ we have that $H(v) = 0$ if and only if $v = \vS.$ Furthermore, since $a$ and $\omega$ are strictly positive, the real parts of the eigenvalues of the matrix in the definition of $H$ are also positive. This can be seen from the following set of observations. To begin with,  the associated characteristic equation of this matrix is 
\[
    \lambda^2 - \lambda \omega + a = 0.
\]
Hence, the roots are $\lambda = (\omega \pm \sqrt{\omega^2 - 4a})/2.$ If $\omega^2 < 4a,$ then the roots are complex valued; therefore, the real part of both these roots is $\omega/2$ which is clearly positive. On the other hand, if $\omega^2 \geq 4a,$ then both the roots are real; further, the smallest of the two roots, i.e., $(\omega - \sqrt{\omega^2 - 4a})/2,$ is strictly positive since $a > 0.$ This shows that the negative of the matrix given in the definition of $H$ is Hurwitz. Together, these observations show that $\vS$ is the unique globally asymptotically stable equilibrium of the ODE $\dot{v}(t) = H(v(t)).$ This verifies Condition iii.).

Finally, let 
\[
H_\infty(v) = -\begin{bmatrix} 
    \omega & a \\ 
    -1 & 0 
    \end{bmatrix} v.
\]
Then, it is easy to see that $H_c(v) \to H_\infty(v)$ uniformly on compact sets as $c \to \infty.$ 
Also, $H_\infty(v) = 0$ if and only if $v = 0.$ Furthermore, as shown before, the negative of the matrix in the definition of $H_\infty$ is Hurwitz. This implies that the origin is the unique globally asymptotically stable equilibrium of the ODE $\dot{v}(t) = H_\infty(v).$ This verifies condition iv.).

It now remains to check if $\{\epsilon_k\}$ has the decaying behaviour described in Proposition~\ref{prop:1TS.Conv.Perturbed}. Towards this, since $|M_{k + 1}| \leq (p + |z_k|),$ we have
\[
    \left\|\begin{bmatrix} 0 \\     \epsilon_k \end{bmatrix} \right\| \leq C' \gamma_k (1 + |u_k| + |z_k|) \leq  C \gamma_k (1 + \|v_k\|)
\]
for some constants $C, C' \geq 0.$ Now, because $\gamma_k$ decays to $0$ as $k \to \infty$ due to the assumption in Statement 1., it follows that $\{\epsilon_k\}$ indeed has the desired behaviour.

This completes the proof in the one-timescale setup.

\emph{Two-timescale  Setup}: Since $\beta_k/\gamma_k \to 0,$ one can perceive  $u_k$ to be changing on a faster timescale relative to $y_k.$ Hence, the update rules in \eqref{e:u_Upd} and \eqref{e:z_Upd} can  be viewed as a two-timescale stochastic approximation. We now verify the conditions listed in Theorem~\ref{thm:2TS.Conv} and then use Proposition~\ref{prop:2TS.Conv.Perturbed} (both given in the appendix) to conclude $z_k \to \Delta$ a.s. 

Conditions i.) and ii.) trivially hold. Hence, we only focus on verifying Conditions iii.) and iv.) Because of linearity, $h$ and $g$ are trivially Lipschitz continuous. Next, let $\phi(z) = (b - az)/\omega$ for $z \in \Real.$ Clearly, $\phi$ is linear in $z$ and, hence, Lipschitz continuous. Also, $h(\phi(z), z) = 0.$ This, along with the fact that the sign in front of $u$ in $h(u, z)$ is negative, shows that $\phi(z)$ is indeed the unique globally asymptotically stable equilibrium of the ODE $\dot{u}(t) = h(u(t), z).$ Next, observe that the ODE $\dot{z}(t) = g(\phi(z(t)), z(t))$ has the form $\dot{z}(t) = (b - a z(t))/\omega.$ Clearly, this ODE  has $\Delta$ as its  unique globally asymptotically stable equilibrium. This completes the verification of Condition iii.).

With regards to Condition iv.), first let $h_\infty$ be the function defined by $h_\infty(u, z) = -\omega u - a z.$ Also, for $z \in \Real,$ let $\phi_\infty(z) = -az/\omega.$ This function is linear in $z$ and, hence, Lipschitz; also, $\phi_\infty(0) = 0.$ Then, on the one hand, $h_c \to h_\infty$ uniformly on compacts as $c \to \infty$ and, on the other hand, the ODE $\dot{u}(t) = h_\infty(u(t), z)) = - \omega u(t) - a z$ indeed has $\phi_\infty(z)$ as its unique globally asymptotically stable equilibrium. Finally, for $z \in \Real,$ let $g_\infty(z) = - az/\omega.$ Then, trivially, $g_c \to g_\infty$ uniformly on compacts, as $c \to \infty.$  Further, $\dot{z}(t) = g_\infty(z(t)) = - a z(t)/\omega$ which indeed has the origin as its unique globally asymptotically stable equilibrium. With this, we finish with verifying Condition iv.).

Now, as per Proposition~\ref{prop:2TS.Conv.Perturbed}, we need to show that $\{\epsilon_k\}$ is asymptotically negligible. However, this is indeed true since $|M_{k + 1}| \leq (z_k + p)$ which implies $|\epsilon_k| \leq C \gamma_k (1 + |u_k| + |z_k|)$ for some constant $C \geq 0,$ and since $\gamma_k \to 0.$

This shows that $(u_k, z_k) \to (\phi(\Delta), \Delta) = (0, \Delta)$ a.s., as desired. 
\end{proof}

\subsection{Comparison with Existing Estimators}
\label{subsec:Comp}
As far as we know, there are three other approaches in the literature for estimating page change rates---the Naive estimator, MLE, and the MM estimator. The details about the first two estimators can be found in \cite{cho2003estimating} while, for the third one, one can look at \cite{upadhyay2019}. We now do a comparison, within the context of our setup, between these estimators and the ones that we have proposed. 

The Naive estimator simply uses the average number of changes detected to approximate the rate at which a page changes. That is, if $\{q_k\}$ denotes the iterates of the Naive estimator then, in our setup, $q_k = p\hI_k/k,$ where $\hI_k$ is as defined below \eqref{eqn:LLN_Est}. The intuition behind this is the following. If  $\tau_1$ is as defined at the beginning of Section~\ref{subsecLLN_Est}, then 
\begin{equation}
\label{e:Exp.N.tau1}
    \E[N(\tau_1)] = \Delta/p.
\end{equation}
Thus, the Naive estimator tries to approximate $\E[N(\tau_1)]$ with $\hI_k/k$ then use \eqref{e:Exp.N.tau1} to determine the change rate.  

Clearly, $ \E[q_k] = p \Delta/(\Delta + p) \neq \Delta.$ Also, from the strong law of large numbers, $q_k \overset{a.s.}{\to} p \Delta/(\Delta + p) \neq \Delta.$ Thus, this estimator is not consistent and is also biased. This is to be expected since this estimator does not account for all the changes that occur between two consecutive accesses.

Next, we look at the MLE estimator. Informally, this estimator identifies the parameter value that has the highest probability of producing the
observed set of observations. In our setup, the value of the MLE estimator is obtained by solving the following equation for $\Delta:$
\begin{equation}\label{LLE}
\sum_{j=1}^{k}  I_j\, \tau_j/(\exp{( \Delta\, \tau_j)} - 1) =  \sum_{j=1}^{k}  (1 - I_j)\, \tau_j,
\end{equation}
where $\tau_k = t_k - t_{k - 1}$ and $\{t_k\}$ is as defined in Section~\ref{sec2}. 
The derivation of this relation is given in \cite[Appendix C]{cho2003estimating}. As mentioned in \cite[Section 4]{cho2003estimating}, the above estimator is consistent.

Note that the MLE estimator makes actual use of the inter-arrival crawl times $\{\tau_k\}$ unlike our two estimators and also the Naive estimator. In this sense, it fully accounts for the information available from the crawling process. Due to this, as we shall see in the experiments section, the quality of the estimate obtained via MLE improves rapidly in comparison to the Naive estimator as the sample size increases.

However, MLE suffers in two aspects: computational tractability and mathematical instability. Specifically, note that the MLE estimator lacks a closed form expression. Therefore, one has to solve \eqref{LLE} by using numerical methods such as the Newton–Raphson method, Fisher’s Scoring Method, etc. Unfortunately, using these ideas to solve \eqref{LLE} takes more and more time as the number of samples grow. Also note that, under the above solution ideas, the MLE estimator works in an offline fashion. In that, each time we get a new observation, \eqref{LLE} needs to be solved afresh. This is because there is no easy way to efficiently reuse the calculations from one iteration into the next (note that the defining equation \eqref{LLE} changes in a significant and nontrivial way from one iteration to the other). 

Besides the complexity, the MLE estimator is also unstable in two situations. One, when no changes have been detected ($I_j = 0, \, \forall k \in \{1, \ldots, k\}$), and the other, when all the accesses detect a change ($I_j = 1, \, \forall k \in \{1, \ldots, k\}$). In the first setting, no solution exists; in the second setting, the solution is $\infty.$ One simple strategy to avoid these instability issues is to clip the estimate to some pre-defined range whenever one of bad observation instances occur. 

Finally, let us discuss the MM estimator. Here, one looks at the fraction of times no changes were detected during page accesses and then, using a moment matching method, tries to approximate the actual page change rate. In our context, the value of this estimator is obtained by solving $\sum_{j = 1}^k (1 - I_j) = \sum_{j = 1}^k e^{- \Delta \tau_j}$ for $\Delta.$ The details of this equation are given in \cite[Section~4]{upadhyay2019}. While the MM idea is indeed simpler than MLE, the associated estimation process continues to suffer from similar instability and computational issues like the ones discussed above. 

We emphasise that none of our estimators suffer from any of the issues mentioned above. In particular, all of our estimators are online and have a significantly simple update rule;  thus, improving the estimate whenever a new data point arrives is extremely easy. Moreover, all of them are stable, i.e., the estimated values will almost surely be finite. More importantly,  the performance of our estimators is comparable to that of MLE. This can be seen from the numerical experiments in Section~\ref{sec4:numericalExpts}.

\subsection{Comparison of Theorem~\ref{thm:SAM_Est} with the Literature on Stochastic Momentum Methods}
\label{s:SAM.Literature}

We first provide an alternative characterization of \eqref{eqn:SAM_Est}. Let $f(z) = \frac{1}{2a}(az - b)^2,$ where $a$ and $b$ are as defined below \eqref{e:1TS.martingale.cond}, and let $h$ be as defined in Section~\ref{subsec:SA}. Then, clearly, $h(z) = -\nabla f(z).$ Thus, \eqref{eqn:SAM_Est} can be rewritten as 
\[
    z_{k + 1} = z_k + \eta_k[-\nabla f(z_k) + M_{k + 1}] + \zeta_k (z_k - z_{k - 1}),
\]
where $M_{k + 1}$ is as defined in \eqref{d:Mart.Diff.Seq}. Consequently, it follows that \eqref{eqn:SAM_Est} can also be viewed as an SGD method with a heavy-ball momentum term (similarly, \eqref{eqn:SA_Est} is also an SGD method, but we will not focus on that here).

The above viewpoint now brings forth an interesting question ``How does Theorem~3 compare with the existing results on stochastic heavy-ball method and the stochastic variant of Nesterov’s accelerated gradient method?"

While there are numerous results on stochastic momentum methods, surprisingly, most of them hold only under extremely restrictive assumptions: they either need
\begin{enumerate}
    \item that the gradient of the objective function be uniformly bounded \cite{gitman2019understanding, yang2016unified}, or
    
    \item that the noise sequence, i.e.,  $\{M_{n + 1}\},$ be independent of the iterates \cite{aybat2020robust} or, alternatively, its variance be uniformly bounded \cite{gitman2019understanding, yang2016unified, aybat2020robust, laborde2020lyapunov, assran2020convergence, kulunchakov2019estimate, can2019accelerated, cohen2018acceleration}.
\end{enumerate} 

In our setup, in contrast, the objective function $f$ is quadratic; hence, the magnitude of its gradient grows to infinity as $|z| \to \infty.$ Also, $\E[|M_{k + 1}|^2|\cF_k] = (z_k + p)^2 \E[I_{k + 1} - \Delta/(\Delta + p)]^2,$ which implies that $\E|M_{k + 1}|^2 = \frac{\Delta p}{(\Delta + p)^2} \E(z_k + p)^2.$ One can thus see that the above assumptions do not directly hold in our case. 

To the best of our knowledge, \cite{vaswani2019fast} and \cite{gadat2018stochastic} are the only other works that similarly do not need the above assumptions. The results in \cite{vaswani2019fast}, however, only apply to the setup with constant stepsizes. In that case, it is shown there that the iterates converge to a neighborhood of the desired solution but not to the solution itself. On the other hand, \cite{gadat2018stochastic} does discuss results on convergence and convergence rates of the stochastic heavy-ball method.
The analysis there, though, does not apply to the stochastic variant of the original heavy-ball method, i.e., the one proposed in \cite[(9)]{polyak1964some}; instead, it applies to a different variant. 

The paper \cite{loizou2020momentum} is one other work on stochastic momentum methods that has recently generated significant attention. However, the results there concern a setup where the objective function is of a different nature to the one we consider here. In particular, instead of the gradient, it is assumed there that the objective function itself is defined via an expectation. 

In this sense, our work is the first to analyze the stochastic heavy-ball method (in its original form) without a priori presuming that the above two conditions hold. As a matter of fact, it is proved in \cite{gadat2018stochastic} that the variant which is considered there cannot be analyzed using the standard ODE based stochastic approximation techniques such as the one proposed in \cite[Chapter~6]{borkar2009stochastic}. Our analysis, in contrast, is able to directly make use of the standard approach. 





\begin{figure}
  \begin{minipage}{.46\linewidth}
    \subfigure[Q-Q Plot: Real Data versus exponential distribution]
      {\includegraphics[width=\linewidth]{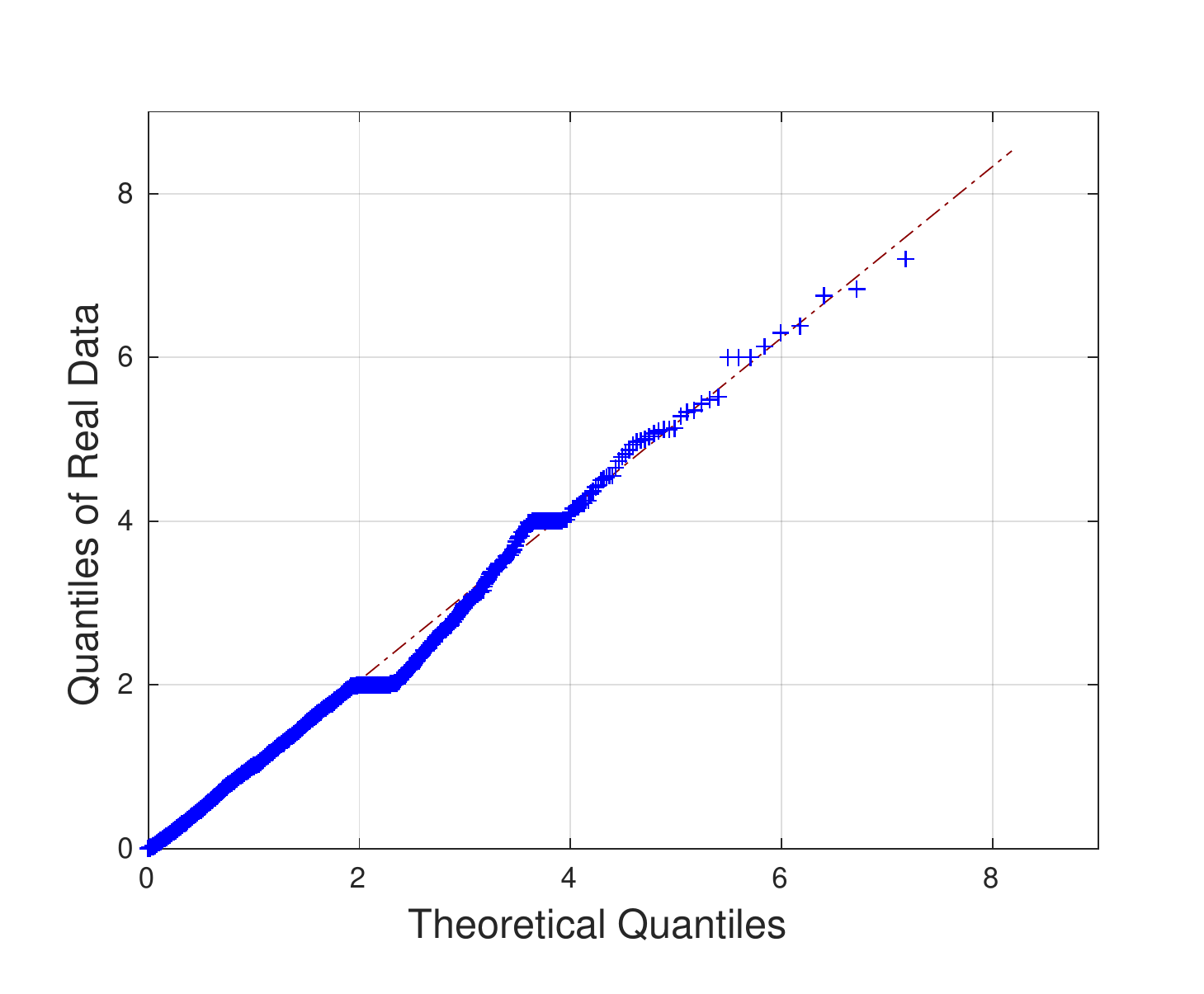}\label{Fig5a}}
     \subfigure[$\Delta = 1.10,\quad p = 0.5$]
      {\includegraphics[width=\linewidth]{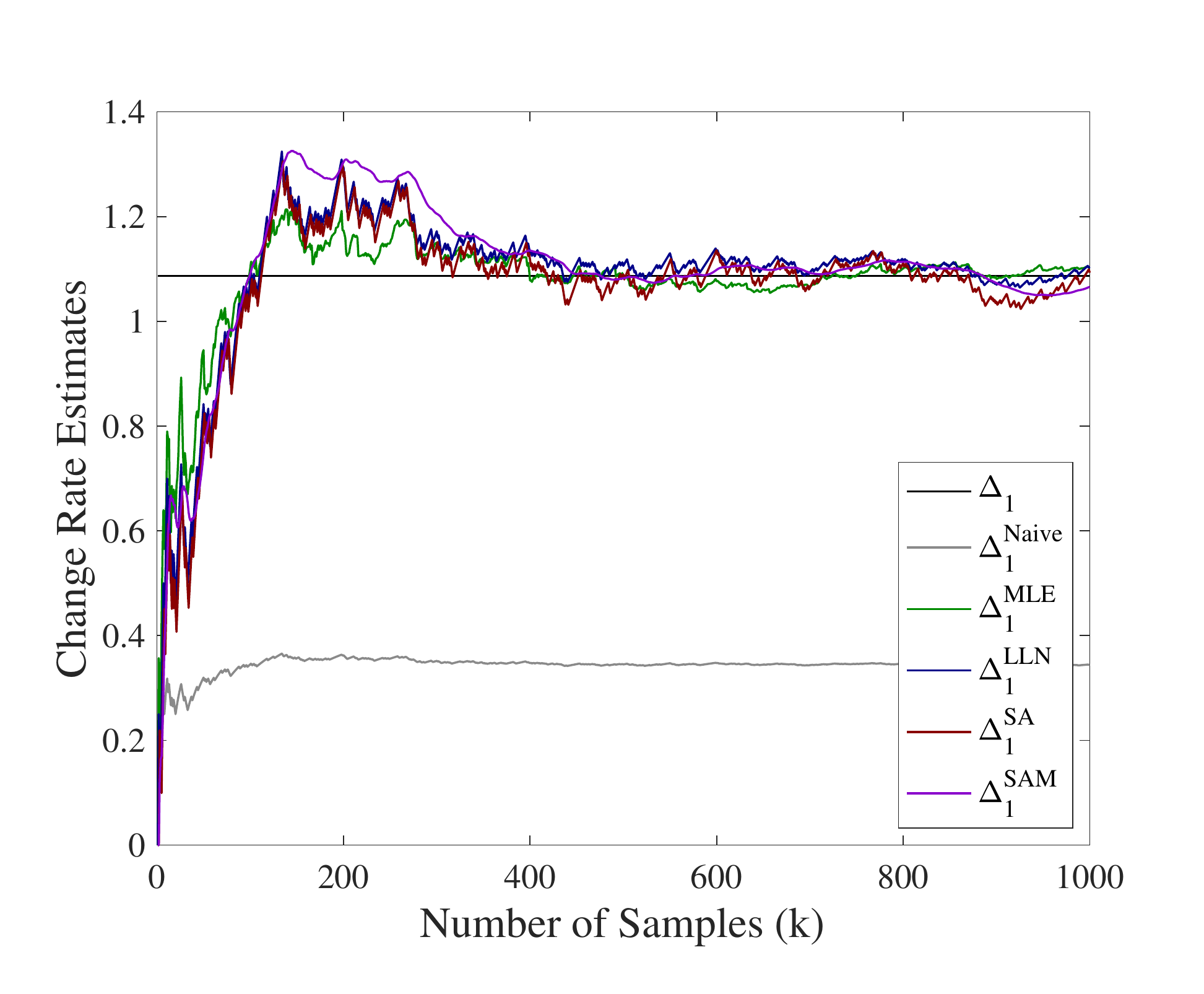}\label{Fig5b}}
     \subfigure[$\Delta = 1.10,\quad p = 0.1$]
      {\includegraphics[width=\linewidth]{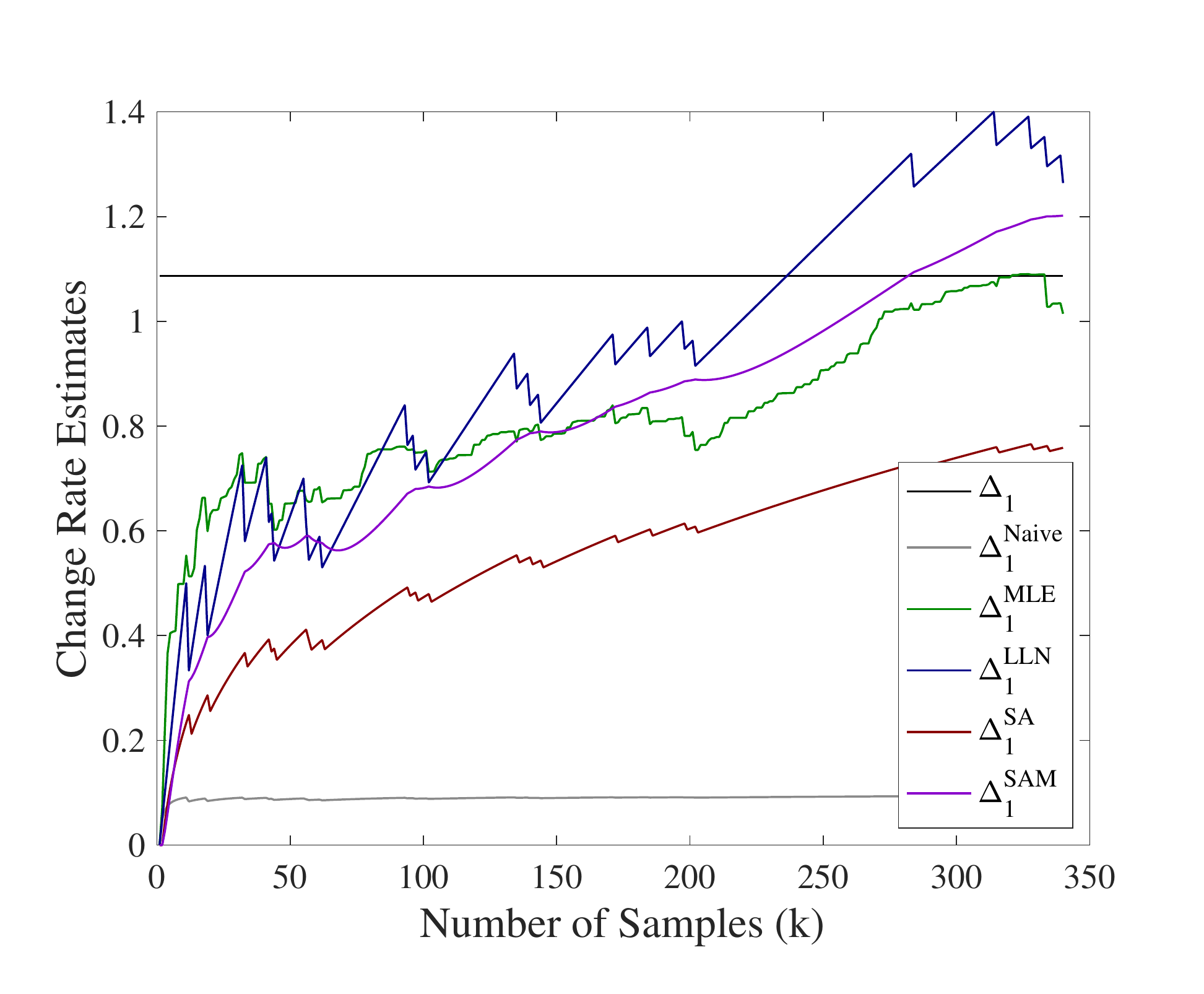}\label{Fig5c}}
    \caption{Different Estimators: Real Data}
    \label{Real Data}
  \end{minipage}\quad
  \begin{minipage}{.46\linewidth}
    \subfigure[Performance of single trajectories]
      {\includegraphics[width=\linewidth]{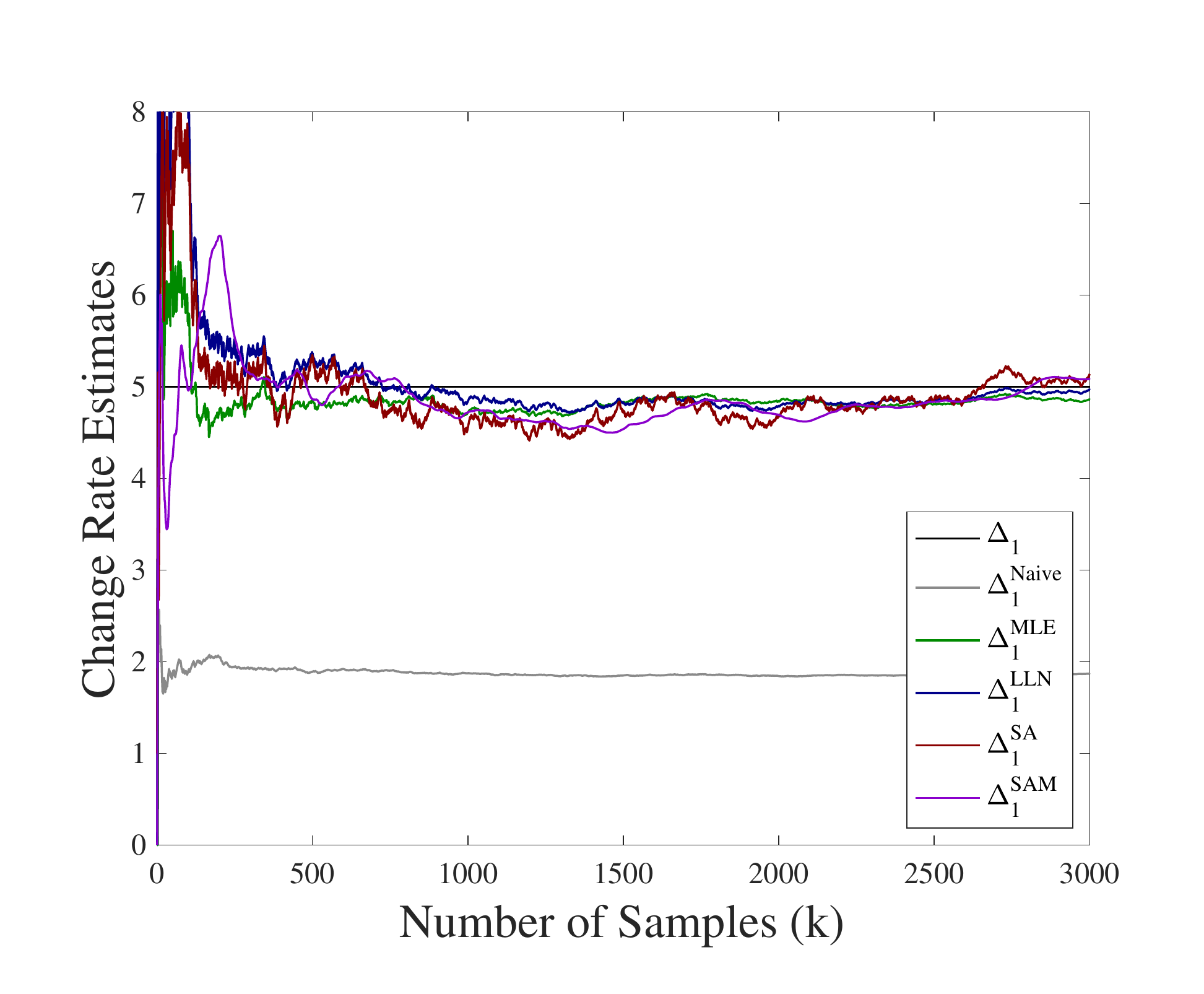}\label{Fig1a}}
     \subfigure[Mean estimate (solid) and $95\%$ Confidence interval]
      {\includegraphics[width=\linewidth]{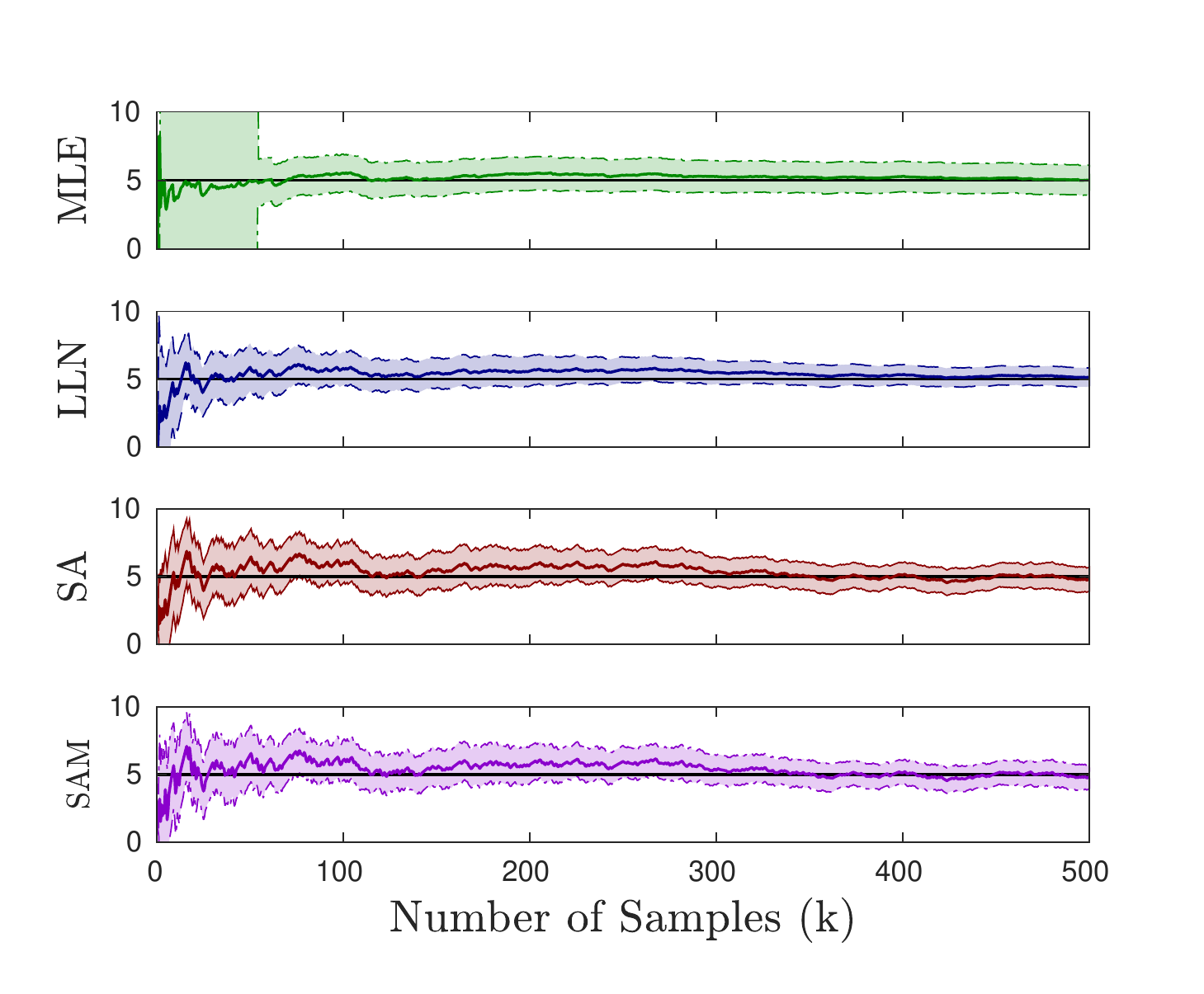}\label{Fig1b}}
     \subfigure[Root mean square error]
      {\includegraphics[width=\linewidth]{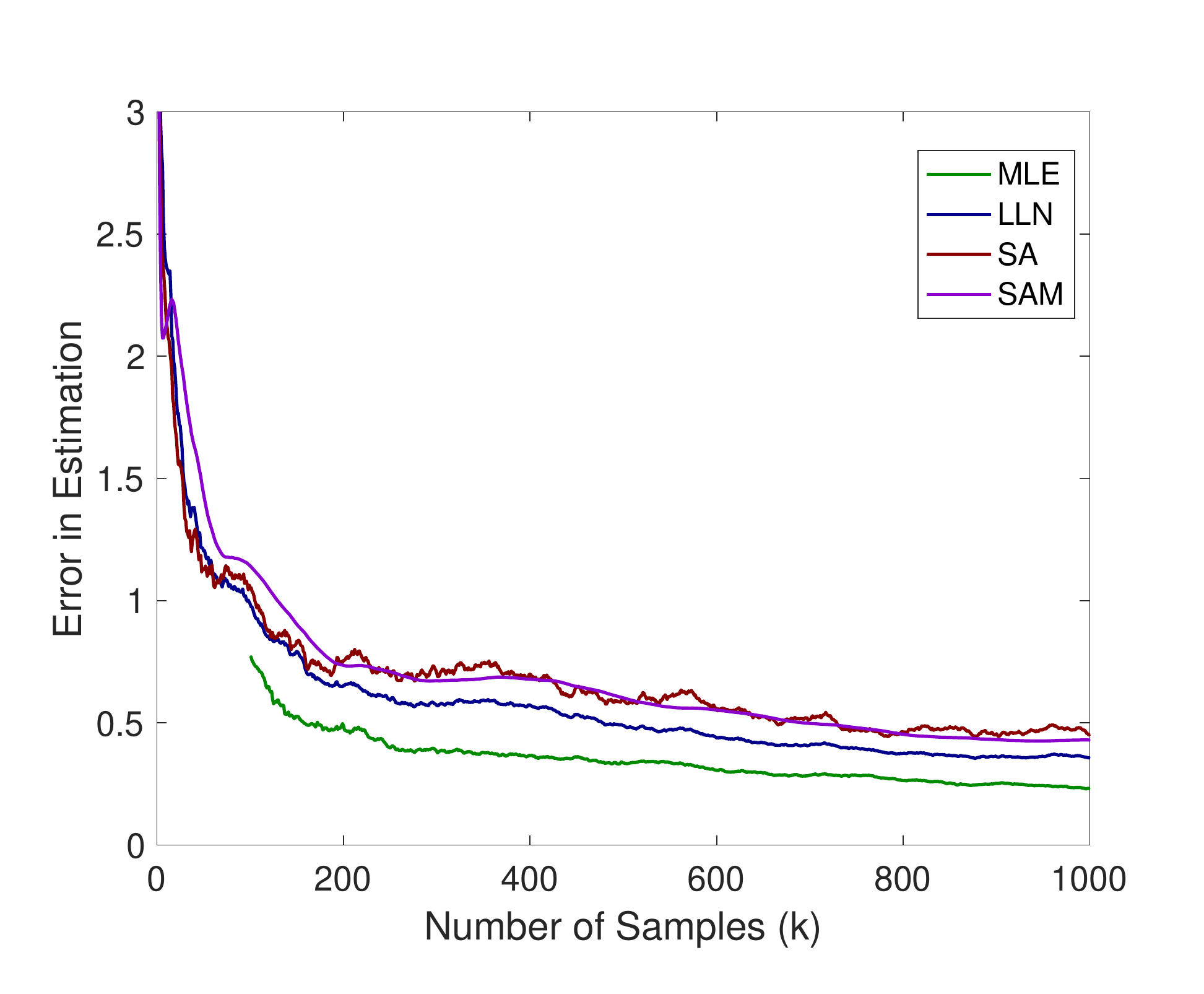}\label{Fig1c}}
    \caption{Synthetic data: $\Delta = 5,\, p = 3$.}
    \label{fig2: change_rates_comparison}
  \end{minipage}
\end{figure}

\section{Numerical Results} 
\label{sec4:numericalExpts}
We now demonstrate the strength of our estimators using three different experiments. The first one involves real data based on Wikipedia traces. It serves two of our goals. First, we use this experiment to validate our model assumption that the page change process is a stationary Poisson point process. Second, we use it to demonstrate that the estimation quality of our online estimators is comparable to that of the offline MLE estimator. In the second experiment, using synthetic data, we study the impact of $\Delta$ and $p$ on our three estimators. In the third experiment, we similarly study how the choices of $\{\alpha_k\},$ $\{\eta_k\}$ and $\{\beta_k\}$ influence the performance. Finally, based on the outcomes of these experiments, we provide some guidelines on which estimator to use in practice. 


\subsection{Performance on Real Data (Expt. 1)}

As mentioned before, our goal here is provide a validation for our model as well as to compare the performance of the different estimators on real data.

To generate the data set, we used Wikipedia traces which are openly available on the web. In particular, we selected an arbitrary page from the list of frequently edited pages on Wikipedia. The title of the page we chose was `Template talk: Did you know". Next, we extracted the timestamps at which this page was edited over a period of five months (April $01,\,2020$ to August $31,\,2020$). We found that this page had changed $4043$ times during this period. From the available history,  we then calculated the inter-update times of the page change process. The average of these values turned out to be $\Delta = 1.1098.$

Using a Q-Q plot, we then compared the distribution (specifically quantiles) of the collected data to that of an exponential distribution with this $\Delta$ rate. The result is given in Fig.~\ref{Fig5a}. Notice that the points roughly fall on a straight line. Importantly, this line is very close to the $45^{\circ}$ diagonal. This implies that both the sets of quantiles come from the same distribution, thereby confirming that the collected inter-update times indeed follow an exponential distribution whose rate is close to $\Delta.$ Equivalently, this implies that the update times come from a Poisson point process with rate close to $\Delta.$


Having verified our assumption, we now compare five different page rate estimators: Naive, MLE, LLN, SA, and SAM. Their performances are given  in Fig~\ref{Fig5b} and Fig~\ref{Fig5c}.

The procedure we adopted to obtain these plots was as follows. (Unless specified, we follow the notations from Section~\ref{sec2}). Recall that we had access to the actual timestamps at which this Wikipedia page was  changed. Keeping this in mind, we artificially generated the crawl instances of this page. These times were sampled from a Poisson point process with rate $p=0.5$ for Fig~\ref{Fig5b} and with $p=0.1$ for Fig~\ref{Fig5c}. We then checked if the page had changed or not between each of the successive crawling instances. This then generated the values of the indicator sequence $\{I_k\}.$ For $p = 0.5,$ the length of this sequence was $1723$ while, for $p = 0.1,$ this length turned out to be $340.$ Using these $I_k,$ $p,$ and inter-update time lengths, we then used the five different estimators mentioned above to find $\Delta.$  This gave rise to the trajectories shown in Fig~\ref{Fig5b} and Fig~\ref{Fig5c}.  Note that the depicted trajectories correspond to exactly one run of each estimator. The trajectory of the estimates obtained by the SA estimator is labeled $\Delta^{SA},$ etc. The stepsizes chosen for our different estimators are as follows.  For our LLN estimator, we had set $\alpha_k \equiv 1$ and, for the SA estimator, we had used $\eta_k = (k + 1)^{-\eta}$ with $\eta = 0.75$. In case of the SAM estimator, we had set $\beta_k = (k + 1)^{-\beta}$ with $\beta = 0.6$ and $\eta_k = (k + 1)^{-\eta}$ with $\eta = 1.2$. (Recall that, in the SAM estimator, the main stepsize is $\eta_k$ while the stepsize multiplying the momentum term has the form $\zeta_k = (\beta_k - \omega \eta_k)/\beta_{k - 1}$).

We now summarise our findings. In Fig~\ref{Fig5b}, we observe that performances of the MLE, LLN, SA and SAM estimators are comparable to each other and all of them outperform the Naive estimator. This last observation is not at all surprising since the Naive estimator completely ignores the changes missed between two successive crawling instances. In contrast to this, we observe that the estimators behave somewhat differently in Fig~\ref{Fig5c}. Recall that the crawling frequency here is $0.1,$ which is quite small compared with the value $0.5$ that was chosen before. We notice that SAM and MLE estimators perform better than SA and LLN estimators in this scenario.



\begin{figure}
  \begin{minipage}{.46\linewidth}
    \subfigure[Performance of single trajectories]
      {\includegraphics[width=\linewidth]{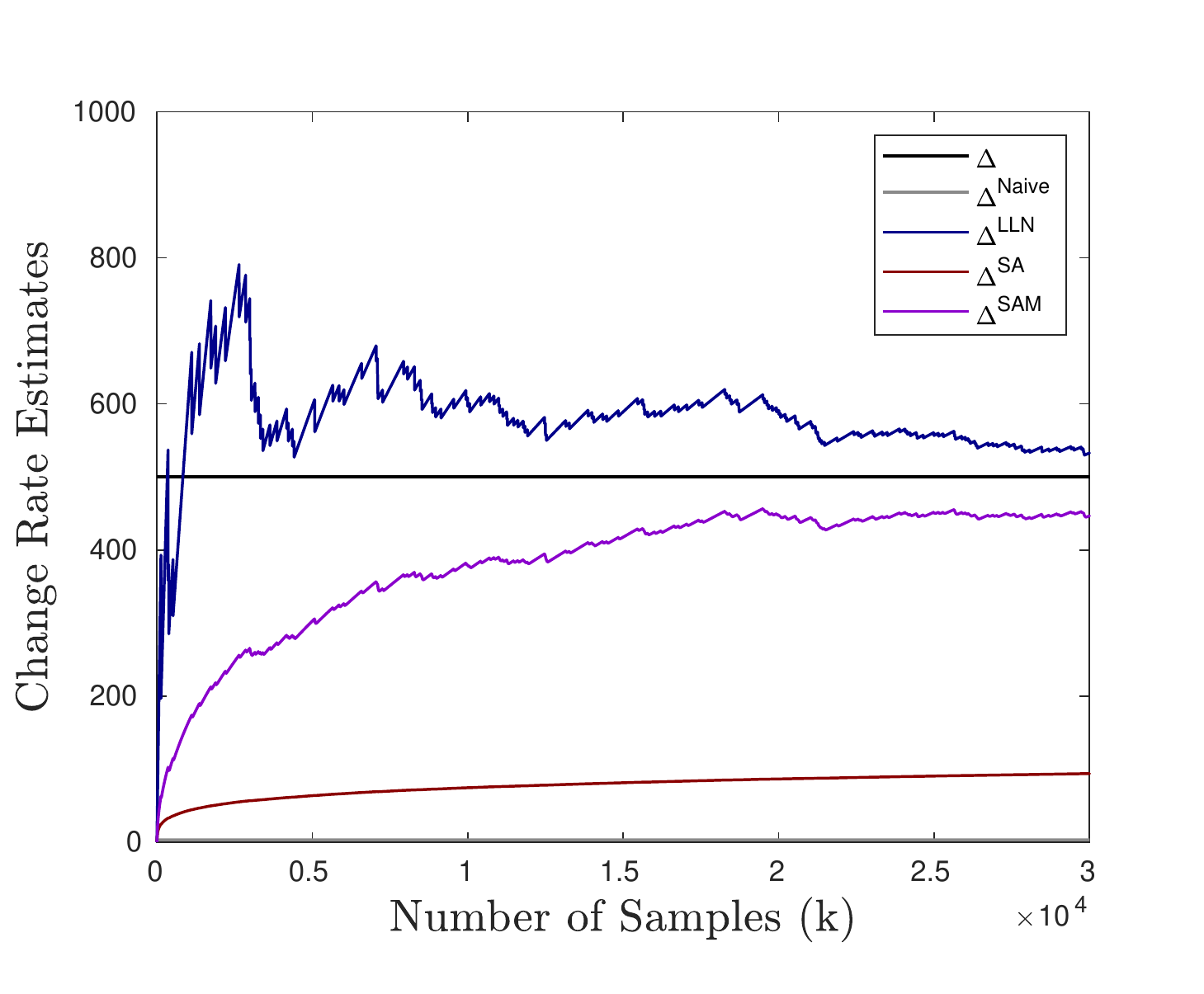}\label{Fig2a}}
     \subfigure[Mean estimate (solid) and $95\%$ Confidence interval]
      {\includegraphics[width=\linewidth]{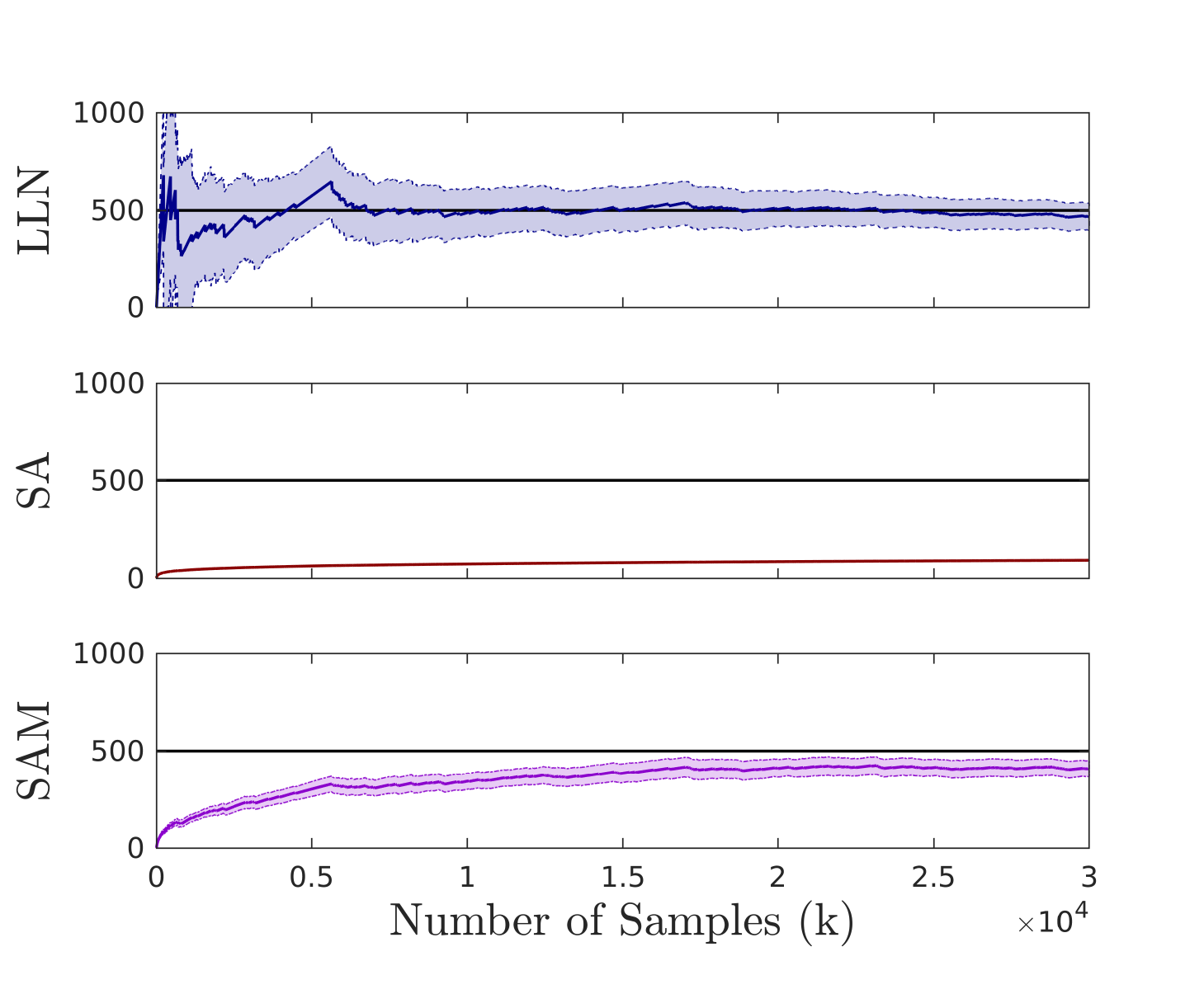}\label{Fig2b}}
     \subfigure[Root mean square error]
      {\includegraphics[width=\linewidth]{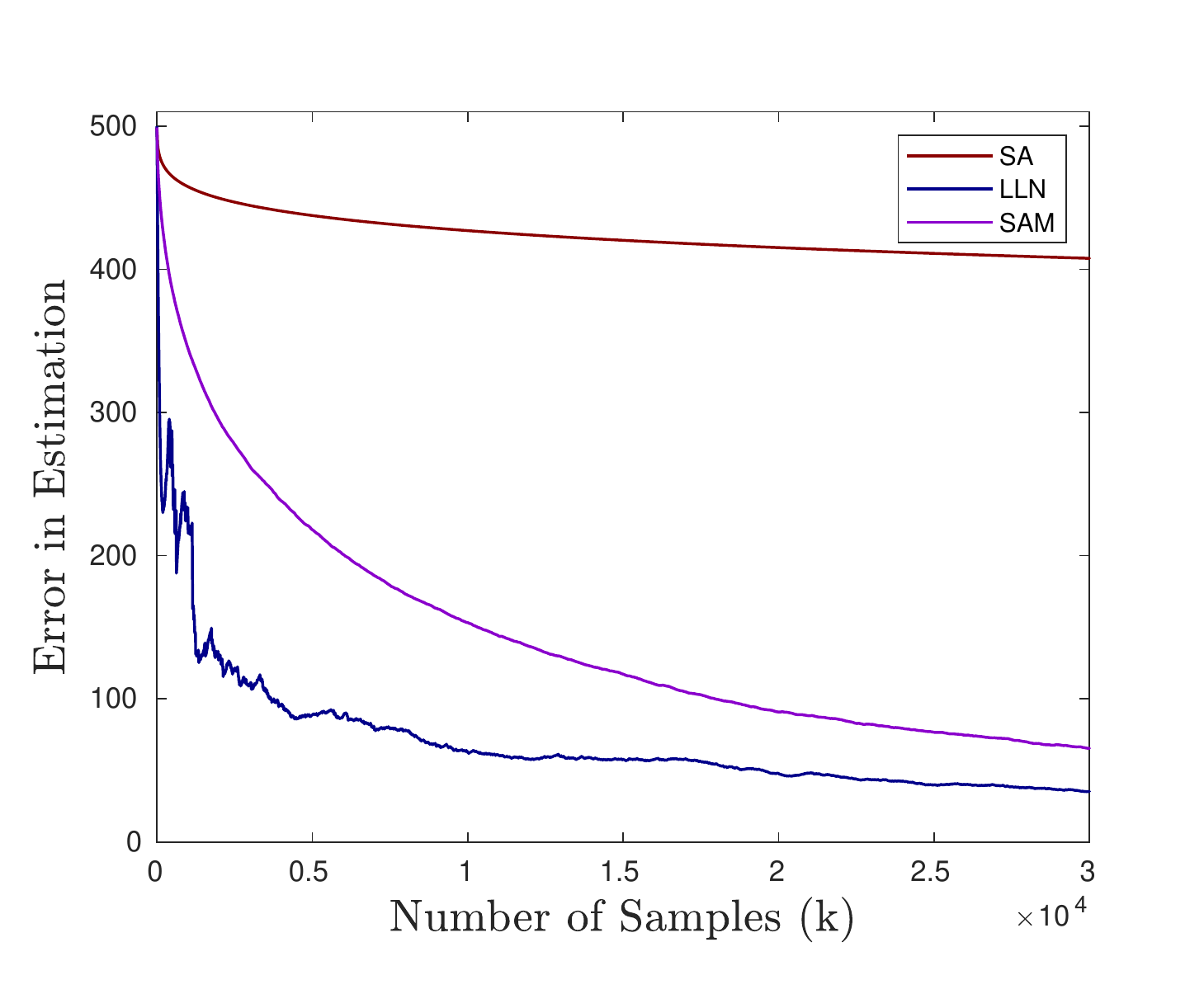}\label{Fig2c}}
    \caption{Synthetic data: $\Delta = 500,\, p = 3$.}
    \label{fig3: change_rates_comparison}
  \end{minipage}\quad
  \begin{minipage}{.46\linewidth}
    \subfigure[Performance of single trajectories]
      {\includegraphics[width=\linewidth]{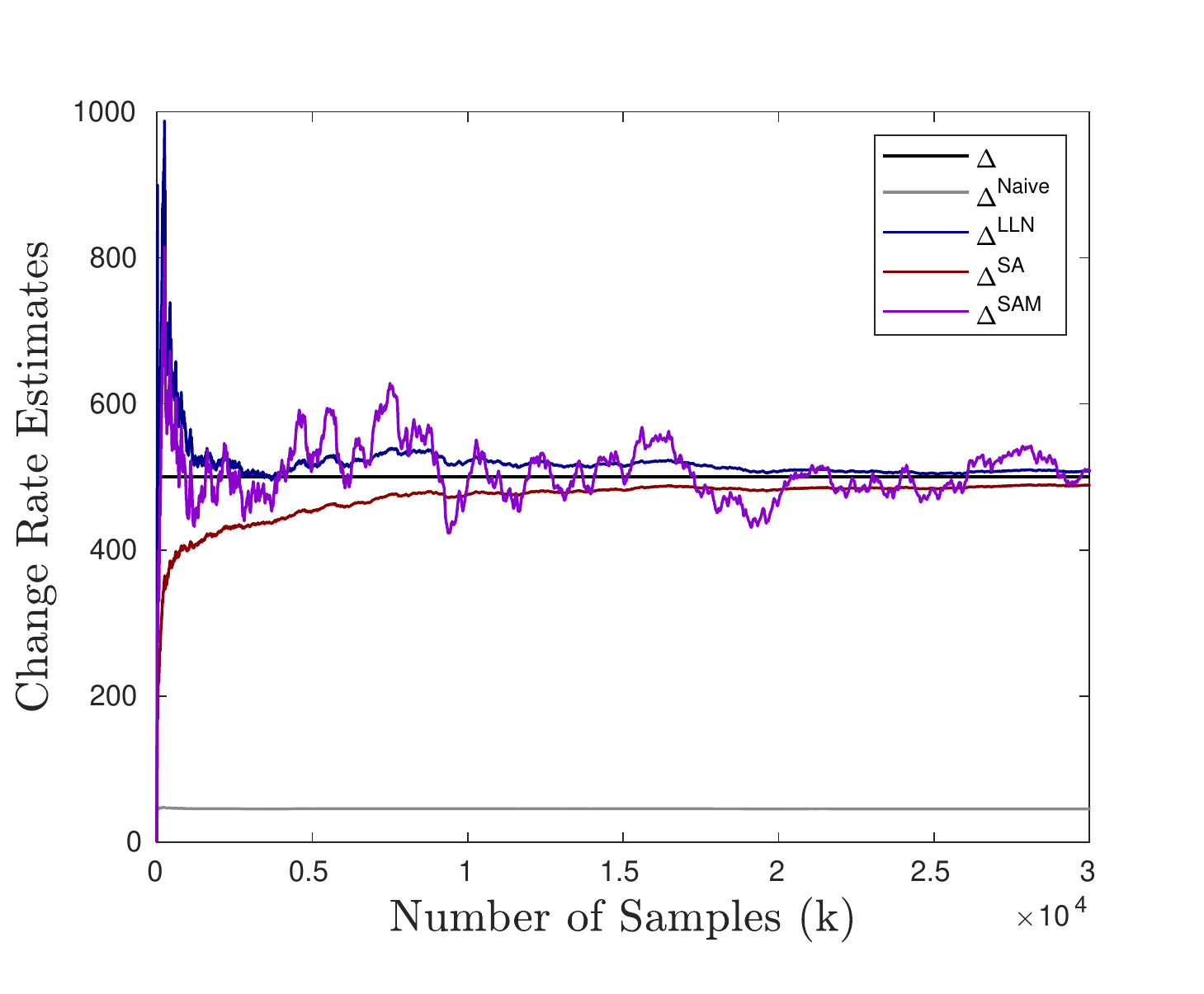}\label{Fig3a}}
     \subfigure[Mean estimate (solid) and $95\%$ Confidence interval]
      {\includegraphics[width=\linewidth]{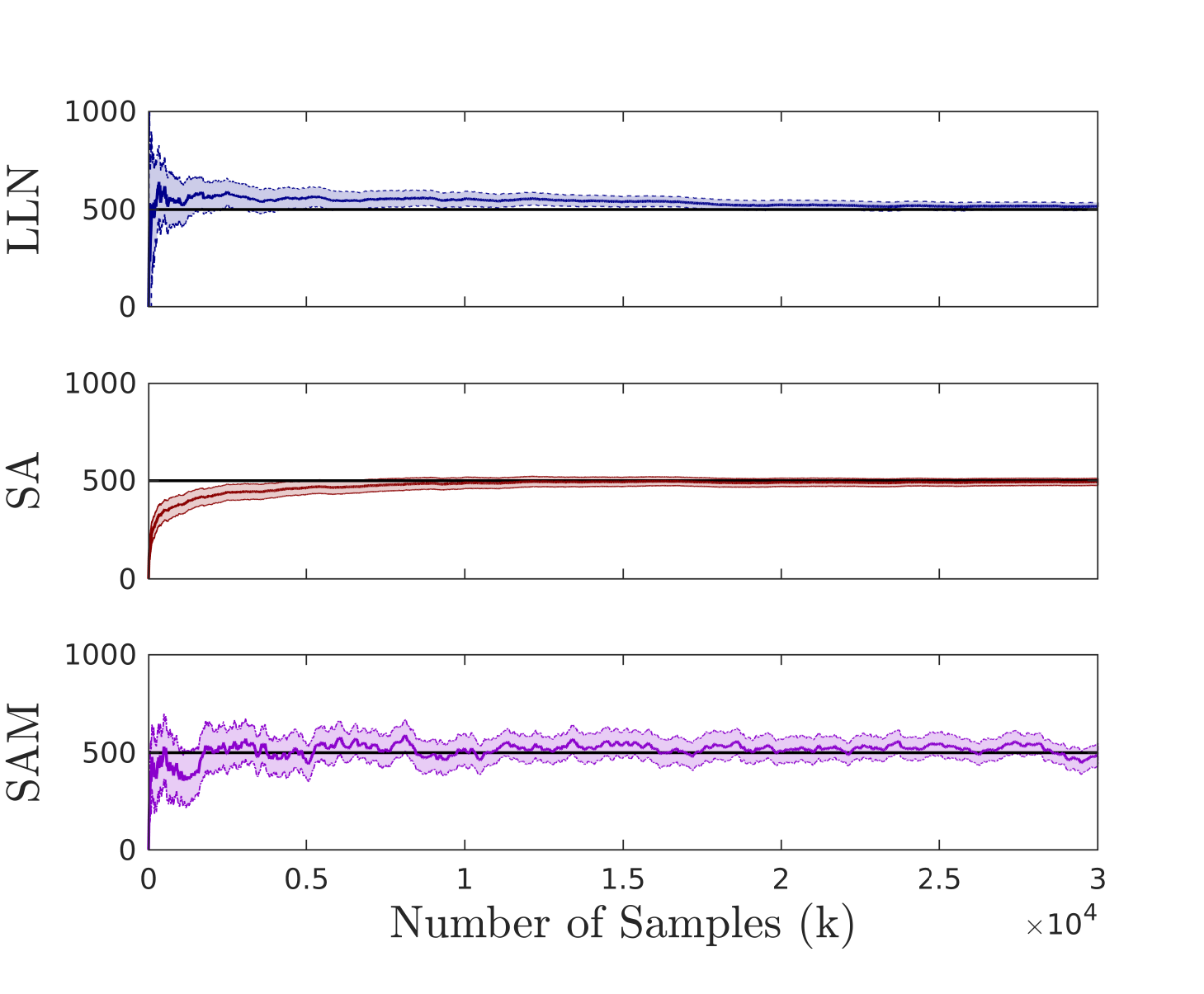}\label{Fig3b}}
     \subfigure[Root mean square error]
      {\includegraphics[width=\linewidth]{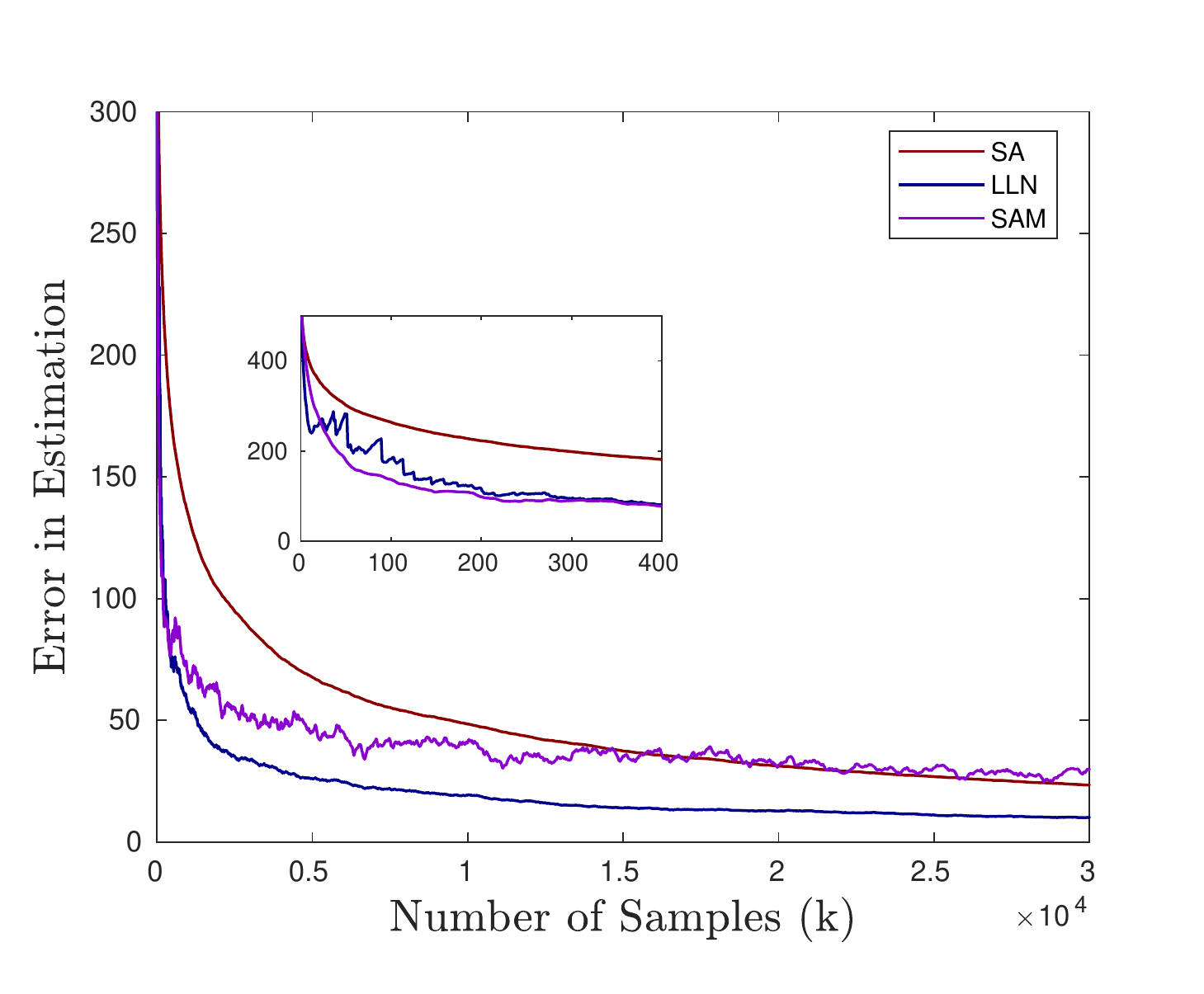}\label{Fig3c}}
    \caption{Synthetic data: $\Delta = 500,\, p = 50$.}
    \label{fig4: change_rates_comparison}
  \end{minipage}
\end{figure} 

\subsection{ Comparison of Estimation Quality using Synthetic Data (Expt. 2)}

Throughout this experiment, we work with synthetic data.

\subsubsection{Sample Variance and Root Mean Squared Error}

Our goal here is to study the sample variance and root mean squared error of the estimates obtained from multiple runs of the different estimators. The output is given in Fig.~\ref{fig2: change_rates_comparison}.

The data for this experiment is generated as follows. We sample points from two different stationary Poisson point processes, one with parameter $\Delta = 5$ and the other with parameter $p = 3.$ We treat the samples from the first process as the times at which an imaginary page changes, and the samples from the second process as the times at which this page is crawled. We then check if the page has changed or not between two successive page accesses. This information is then used to generate the values of the indicator sequence $\{I_k\}.$ 

We now give $\{I_k\},$ $p,$ as well as the inter-access lengths as input to the five different estimators mentioned before. The stepsizes we use are as follows.  For our LLN estimator, we set $\alpha_k \equiv 1$; for the SA estimator, we use $\eta_k = (k + 1)^{-\eta}$ with $\eta = 0.75$; and, for the SAM estimator, we choose $\zeta_k = (\beta_k - \omega \eta_k)/\beta_{k - 1},$ where $\eta_k = (k + 1)^{-\eta}$ with $\eta = 1.3$, $\omega = 1,$ and $\beta_k = (k + 1)^{-\beta}$ with $\beta = 0.75$. Fig.~\ref{Fig1a} depicts one single run of each of the five estimators. 

In Fig.~\ref{Fig1b} and Fig.~\ref{Fig1c}, the parameter values are exactly the same as in Fig.~\ref{Fig1a}. However, we now run the simulation $100$ times; the page change times and the page access times are generated afresh in each run. Fig.~\ref{Fig1b} depicts the $95\%$ confidence interval of the obtained estimates, whereas Fig.~\ref{Fig1c} shows the root mean squared value of the difference between the estimated value and actual change rate of the page.

We now summarize our findings. Clearly, in each case, we observe that performances of the MLE, LLN, SA and SAM estimators are comparable to each other and all of them outperform the Naive estimator. The fact that the estimates from our approaches are close to that of the MLE estimator was indeed quite surprising to us. This is because, unlike MLE, our estimators completely ignore the actual lengths of the intervals between two accesses. Instead, they use $p,$ which only accounts for the mean interval length. Note that the root mean square error   of the first few samples for MLE is very high (hence, it is not depicted in Fig.~\ref{Fig1c}). This is due to the instability that MLE faces; see  Section~\ref{subsec:Comp}. 
Fig.~\ref{Fig1c} shows that the error in the MLE estimate decays faster as compared to others. We believe this is because the MLE also uses the actual interval lengths in its computation; thus, it uses more information about the crawling process than the other estimators. 

While the plots do not show this, we once again draw attention to the fact that the time taken by each iteration in MLE rapidly grows as $k$ increases. In contrast, our estimators take roughly the same amount of time for each iteration.

\subsubsection{Impact of $\Delta$ and $p$ on Performance} 

In the previous experiments, recall that our different estimators more or less behaved similarly. Our goal now is to vary the values of $\Delta$ and $p$ and see if there are any major differences that crop up in their performances. Alongside, we also wish to see the usefulness of the momentum term used in the SAM estimator. The performances in two such interesting scenarios are shown in Fig.~\ref{fig3: change_rates_comparison} and Fig.~\ref{fig4: change_rates_comparison}. Note that we no longer consider  MLE on account of their impractical run times when the $\{I_k\}$ sequence lengths are large. 

In Fig.~\ref{fig3: change_rates_comparison}, $\Delta = 500$ and $p = 3,$ which means the crawling frequency is quite low compared to the frequency at which the page is updated. On the other hand, in Fig.~\ref{fig4: change_rates_comparison}, $\Delta = 500$ and $p = 50;$ thus, the crawling frequency now is relatively higher. The stepsizes for our different estimators are as follows. For the LLN estimator, we chose  $\alpha_k \equiv 1$; for the SA estimator, we chose $\eta_k = (k + 1)^{-\eta}$ with $\eta = 0.8$; and, for the SAM estimator, we chose $\eta_k$ as before, $\omega = 1$ and $\beta_k = (k + 1)^{-\beta}$ with $\beta = 0.5$ (note that our stepsize choice for the SAM estimator violates the conditions of Theorem~\ref{thm:SAM_Est}, but it satisfies the one we made in the conjecture below \eqref{eqn:SAM_Est}).

Fig.~\ref{Fig2a} and  Fig.~\ref{Fig3a} show one  single trajectory of our estimators in the two scenarios. We observe that the LLN and SAM estimators perform quite well as compared to the SA estimator in both the scenarios; however, the latter catches up when the $p$ value becomes higher. The impact of the momentum term can also be clearly seen in the low frequency crawling case. In this scenario, note that the crawler will more or less  always detects a change. That is, the $\{I_k\}$ sequence will mostly consists of all $1$s. In turn, this means that the SA estimator's update rule will almost always have the form $y_{k + 1} = y_k + \eta_k p.$

We then ran the simulation $100$ times and obtained a plot of the $95\%$ confidence interval and the root mean squared error of our different estimators in the two scenarios. This is shown in  Fig.~\ref{Fig2b}, \ref{Fig2c}, \ref{Fig3b}, and \ref{Fig3c}. We observe that variance for SA is relatively very low. This is because the SA estimator does not deviate too much from the update rule mentioned in the previous paragraph. The disadvantage, however, is that its estimates typically are quite far away from the actual change rate. Furthermore, this error decreases quite slowly. Another interesting observation from  Fig.~\ref{Fig2b} and \ref{Fig2c} is that the variance of LLN estimator is larger than that of SAM estimator, however, its error decays at much faster rate than that of the SAM estimator.

Compared to Fig.~\ref{fig3: change_rates_comparison}, notice that in Fig.~\ref{fig4: change_rates_comparison} that performance of all our estimators improve . However, as shown in Fig.~\ref{Fig3b}, the SAM estimator is quite volatile now. Separately, the zoomed-in plot in \ref{Fig3c} shows that the average error for the SAM estimator drops quite rapidly compared to others in the initial few iterations. However, this advantage disappears after $400$ iterations; then on the LLN estimator performs much better. 

\subsection{Impact of Step Size Choices (Expt.3)}
The theoretical results presented in Section~\ref{sec3} show that the convergence rates of LLN, SA, and SAM estimators are affected by the choice of $\{\alpha_k\}$ $\{\eta_k\},$ and $\{\zeta_k\},$ respectively.  Figure~\ref{figv} provides a numerical verification of the same. The details are as follows.  We chose $\Delta = 500$ and $p = 10.$ Notice that the page change rate is again very high, whereas the crawling frequency is relatively very low value. We then use the LLN estimator with three different choices of $\{\alpha_k\};$ these choices are shown in the Fig~\ref{Fig4a} itself. The LLN estimator with $\alpha_k = k^{0.75}$ has the worst performance. This behavior matches the prediction made by Theorem~\ref{thm:LLN_Est}. In Fig.~\ref{Fig4b}, we again consider the same setup as above. However, this time we run the SA estimator with three different choices of $\{\eta_k\};$ the choices are given in the figure itself. We see that the performance for $\eta = 0.5$ is better than the other cases. 

We now analyze the impact of varying $\{\eta_k\}$ and $\{\zeta_k\}$ on the performance of the SAM estimator. Let $\zeta_k$ be of form given in \eqref{d:zeta_k}. Based on our conjecture below \eqref{eqn:SAM_Est}, pick $\eta_k = (k + 1)^{-\eta}$ and $\beta_k = (k + 1)^{-\beta}$ with $\beta \in (0, 1]$ and $\beta < \eta < 2\beta.$ In Fig.~\ref{Fig4c}, we fix $\eta = 0.8$ and vary $\beta$; these choices are shown in the figure itself. The SAM estimator with $\beta = 0.4$ reaches the limit very quickly, however, it is very noisy and keeps fluctuating around actual change rate. The fluctuations reduce as the value of $\beta$  increases; however, larger values of $\beta$ also slow down the rate at which the error decreases. We observe that the SAM estimator with $\beta  = 0.6$ has the best performance. In Fig.~\ref{Fig4d}, we fix $\beta = 0.6$ and vary $\eta$. The figure seems to suggest that a  larger $\eta$ increases the convergence rate but, simultaneously, also increases the fluctuations. 

\begin{figure}[ht!]
\centering
\subfigure[LLN estimator for different $\{\alpha_k\}$ choices] {\includegraphics[width=0.5\textwidth]{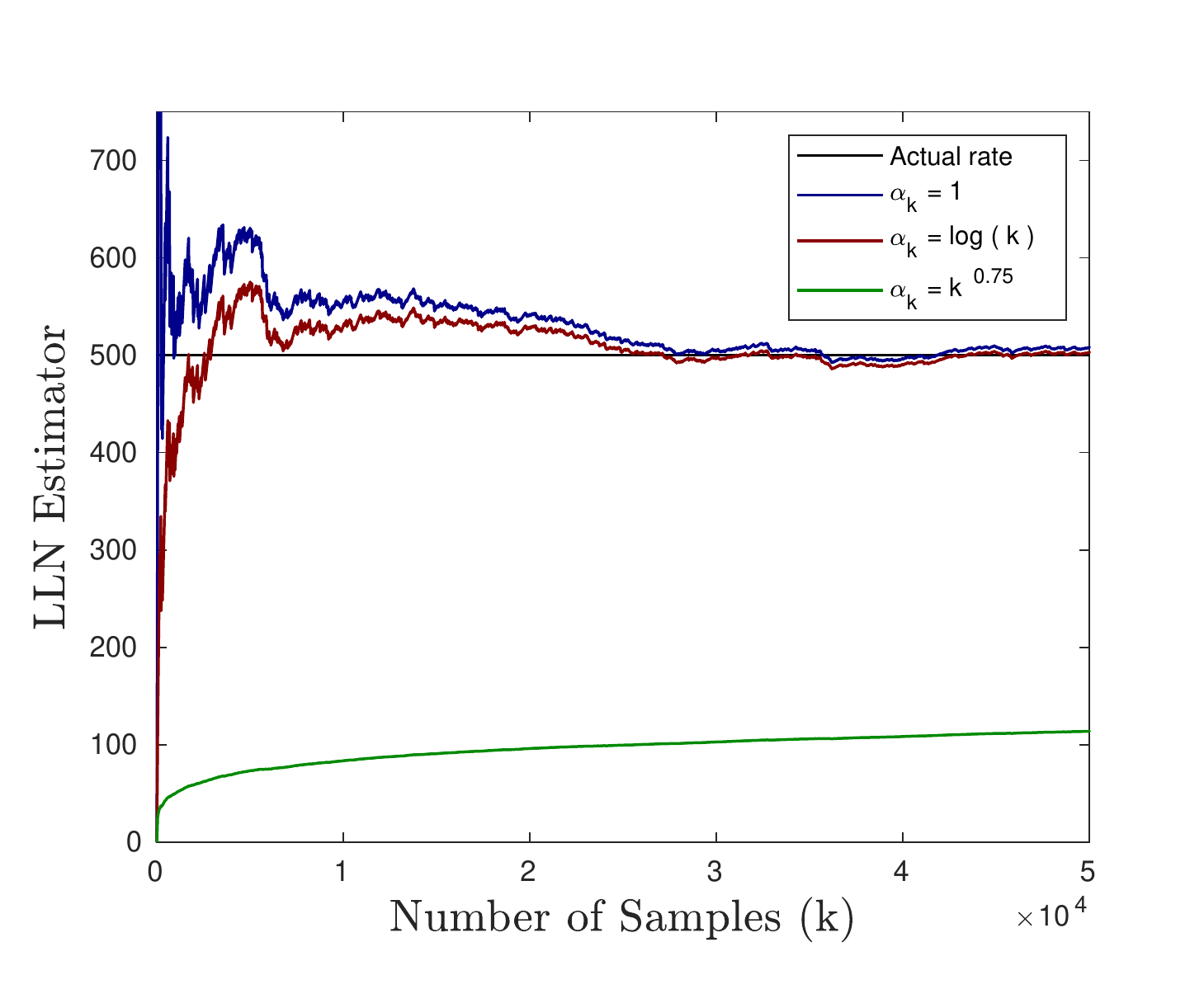}\label{Fig4a}}\hspace{-1em}
\subfigure[SA estimator with $\eta_k = (k+1)^{-\eta} $for different $\eta$ choices] {\includegraphics[width=0.5\textwidth]{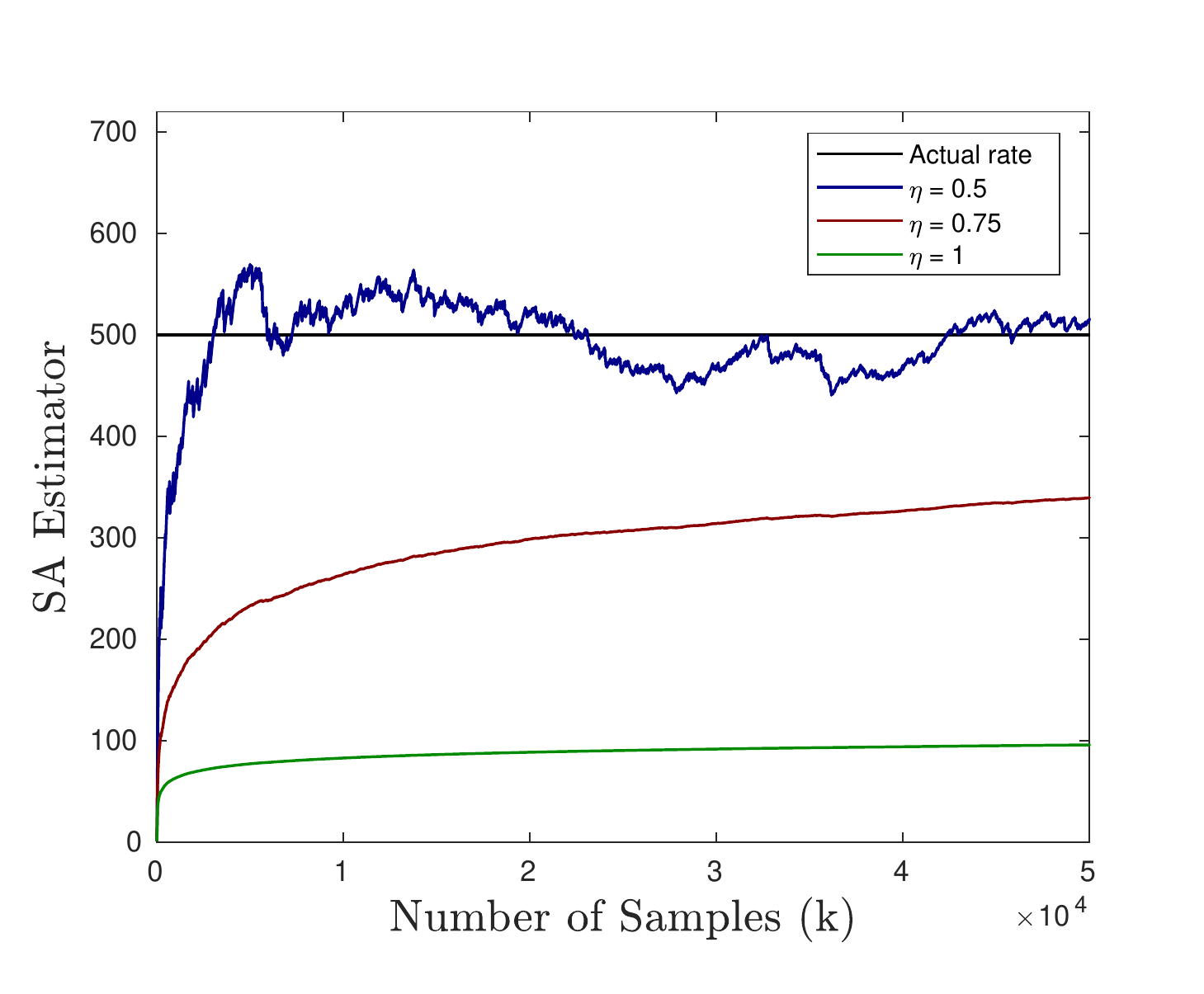}\label{Fig4b}}
\subfigure[SAM estimator with $\eta_k = k^{-0.8}$ for different $\{\beta_k\}$ choices] {\includegraphics[width=0.5\textwidth]{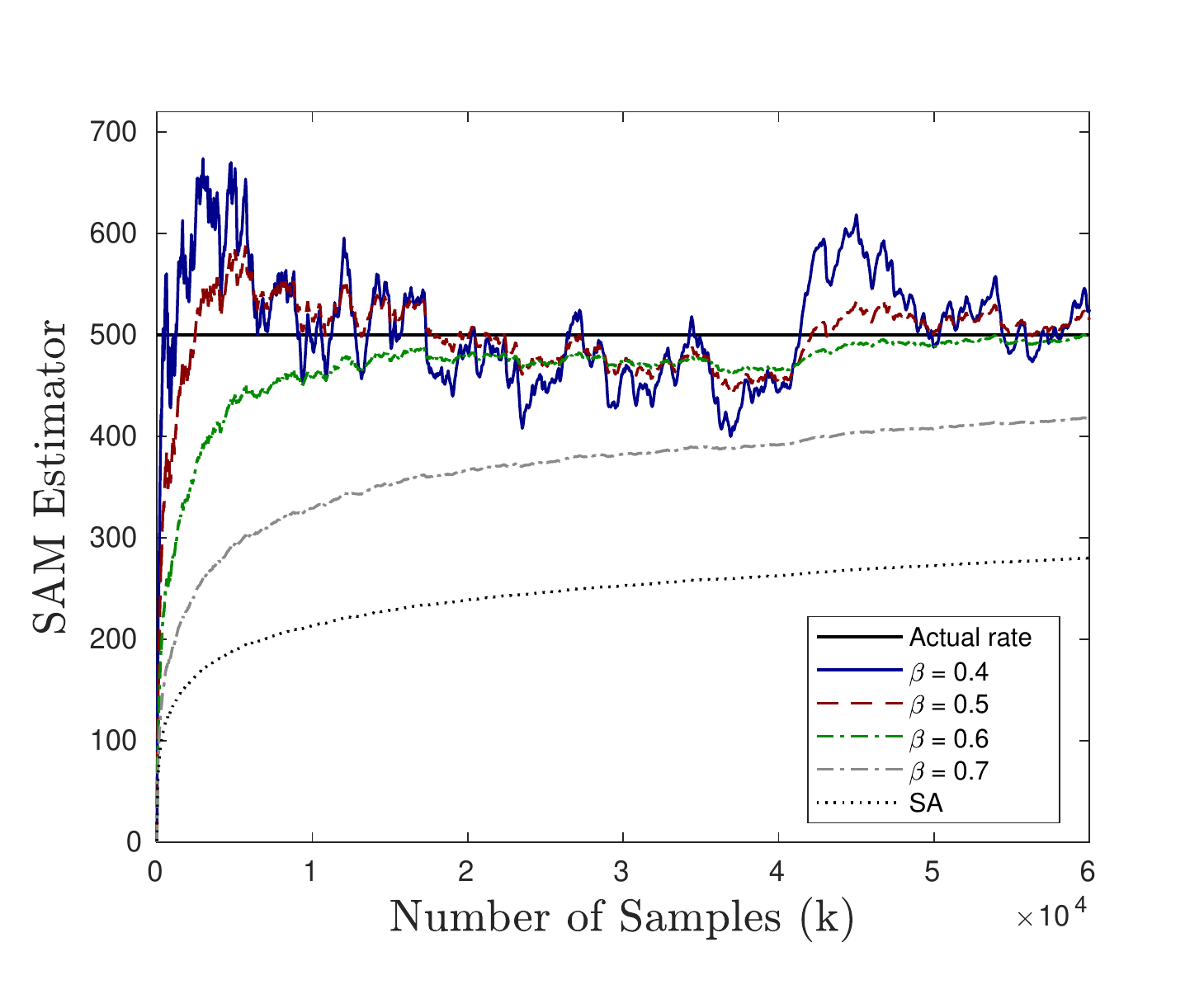}\label{Fig4c}}\hspace{-1em}
\subfigure[SAM estimator with $\beta_k = k^{-0.6}$ for different $\{\eta_k\}$ choices] {\includegraphics[width=0.5\textwidth]{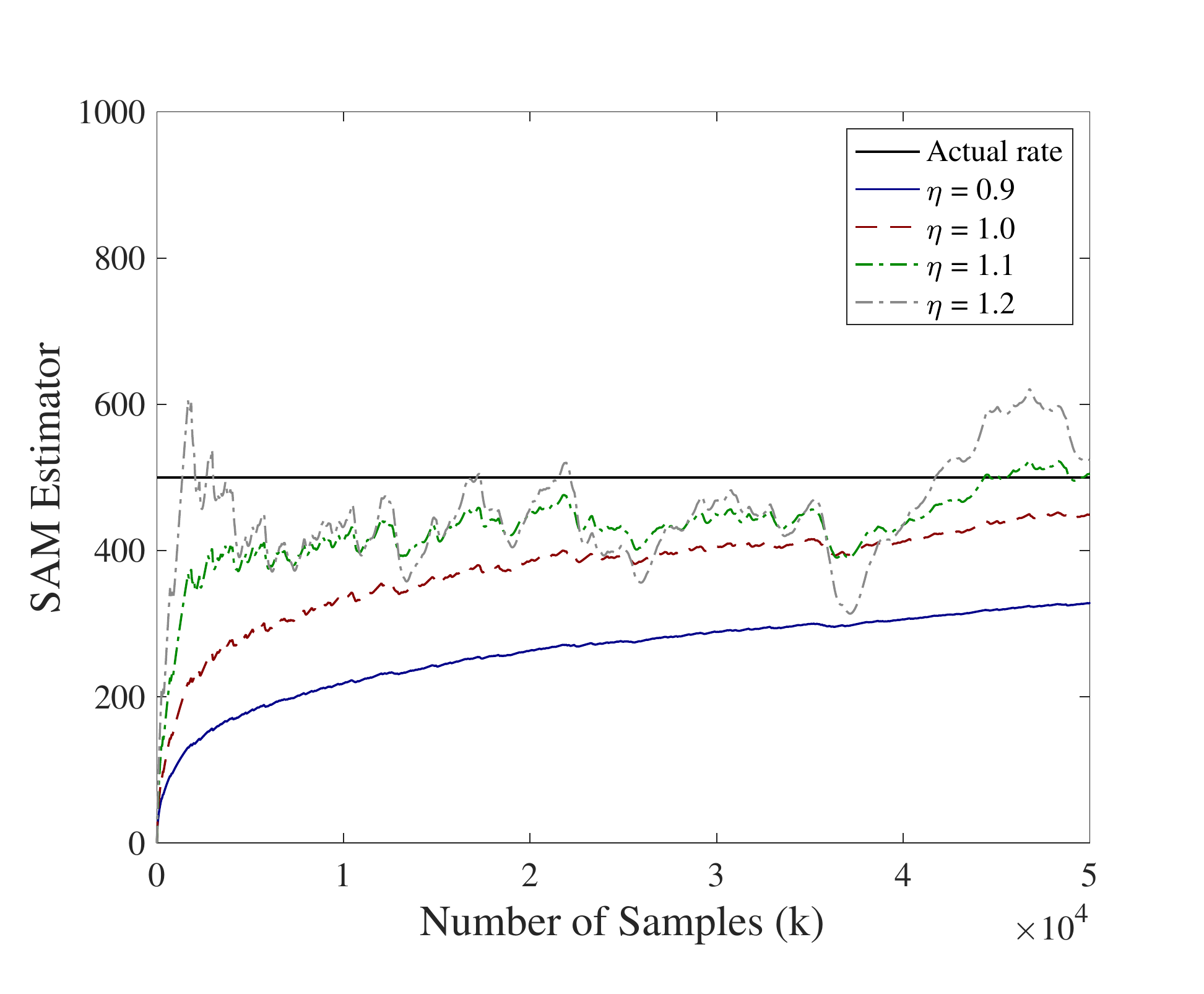}\label{Fig4d}}
\caption{Impact of $\{\alpha_k\}$, $\{\eta_k\}$ and $\{\zeta_k\}$ choices on Performance; $\Delta = 500$ and $p = 10$.}
\label{figv}
\end{figure}

\subsection{Practical Recommendations}
\label{subsec:Practical.Reco}

Here, we provide some recommendations on which estimator to use in practice. Our conclusions are based on what we have observed in the numerical experiments discussed in Section~\ref{sec4:numericalExpts}. We summarise them as follows.
\begin{itemize}
    \item \textit{High frequency crawling:} If the crawling frequency $p$ is comparable to $\Delta$, all estimators (LLN, SA, SAM and MLE)  perform well except the Naive estimator. However, we do not recommend MLE as it is offline and very time-consuming. The examples that correspond to this scenario are depicted in  Fig.~\ref{Fig5b} and Fig.~\ref{fig2: change_rates_comparison}. 
    
    \item \textit{Low frequency crawling:} There are two sub-cases depending on the value of $p$ as compared to $\Delta$. 
    \begin{itemize}
        \item Relatively very low $p$: The Naive estimator is very bad for this scenario as there will several missed changes which will be unaccounted for. We recommend LLN or SAM estimator as they both outperform SA estimator; the example that corresponds to this scenario is depicted in Fig.~\ref{fig3: change_rates_comparison}. For similar reasons as in the previous case, we do not recommend the MLE estimator. 
    
        \item Relatively moderate $p$: The Naive estimator is again a bad choice here. Amongst the rest, we recommend the LLN estimator when several $I_k$ values are available. Otherwise, one can use SAM or the MLE estimator; the offline nature of the MLE will be of concern here as well. The examples that corresponds to this scenario are depicted in  Fig.~\ref{Fig5c} and Fig.~\ref{fig4: change_rates_comparison}.
    \end{itemize} 
\end{itemize}

\section{Estimating Optimal Crawling Rates}
\label{sec:Motivation}
In this section, we discuss how our estimators can be used to identify the optimal crawling rates. Formally, we  suppose that a search engine's local cache consists of $N$ pages. Let $p_i$ denote the rate at which page $i$ is crawled. The goal then is to find the optimal crawling rates such that the overall freshness of the local cache, i.e., 
\begin{equation}\label{Fresh}
    \lim_{T\to\infty}\E\bigg[\dfrac{1}{T} \int\displaylimits_{0}^{T}\bigg( \sum_{i = 1}^{N}w_i \ind{\text{Fresh}(i,t)}\bigg) dt\bigg],
\end{equation}
is maximized subject to the constraint $\sum_{i = 1}^n p_i \leq B$.  Here, $T > 0$ is the time horizon, $w_i$ denotes the importance of the $i$-th page,  $B \geq 0$ is a bound on the overall crawling frequency, $\ind{\text{Fresh}(i, t)}$ is the indicator that page $i$ is fresh at time $t,$ i.e., the local copy matches the actual page.

In \cite{Azar2018}, it was shown that maximizing \eqref{Fresh} under a bandwidth constraint for large enough $T$ corresponds to maximizing $ F(p) = \sum_{i=1}^{N}\big(w_i p_i/(p_i+\Delta_i)\big),$ where $p \equiv (p_1, \ldots, p_N)$. Importantly, it was shown there that this latter optimization problem can be solved efficiently (in $O(N \log N)$ iterations) and provided an algorithm for the same. However, that algorithm requires that the $\Delta_i$'s be known in advance. Our goal here is to combine their algorithm with our estimators  and try and determine the optimal crawling rates.

\begin{figure}
\subfigure[Frequently changing page] {\includegraphics[scale=0.4]{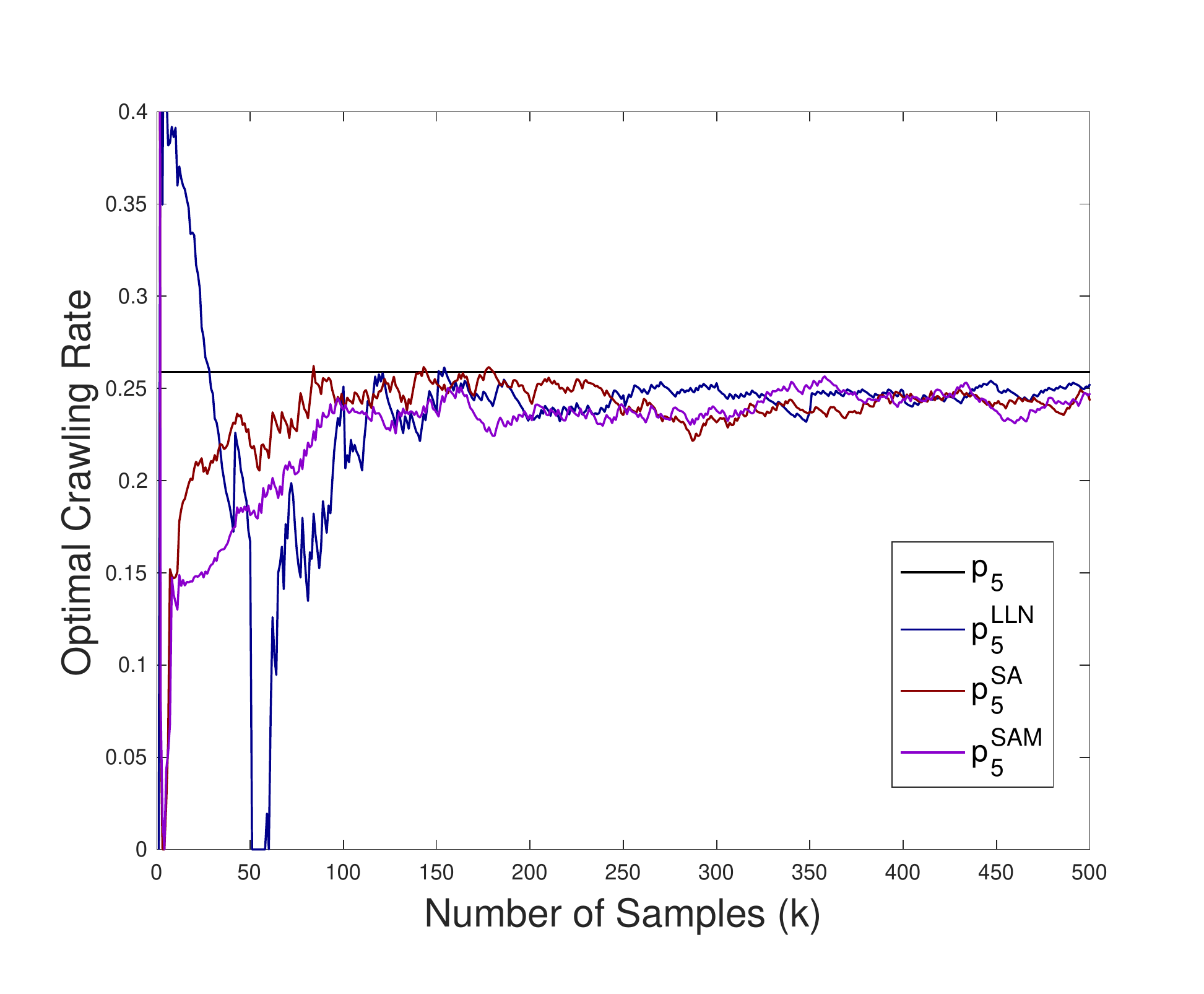}\label{page1H}}
\subfigure[Slowly changing page] {\includegraphics[scale=0.4]{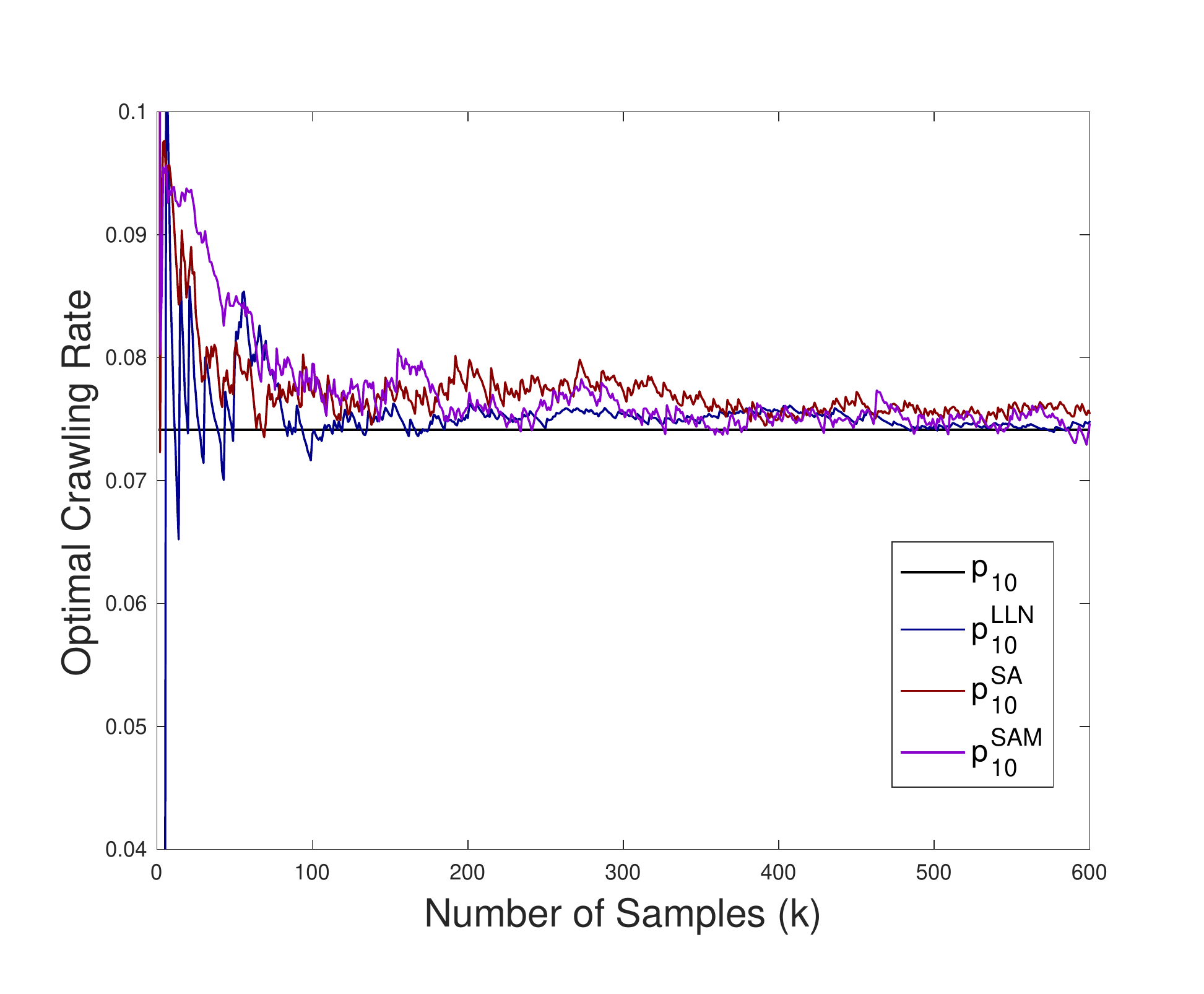}\label{page2L}}
\caption{Adaptive estimation of the optimal crawling rate}
\label{fig: crawl_rates_comparison}
\end{figure} 



Taking inspiration from \cite{andrews2020tracking}, we consider the following hypothetical setup. We consider $N = 50$ pages, in which we presume that there are $\sqrt{50}\approx 7$ pages that change very frequently, i.e., they account for (say) $90\%$ of the total changes in the system. Accordingly, we suppose that the change rate $\Delta_i$ for each frequently changing page is $4.5/7,$ while for the others it is $0.5/43.$ We further assume that the bound on the overall bandwidth is $B = 5.$ We further assume that the frequently changing pages are more important and assign uniform weight of $2$. On the other hand every other page presumed to have uniform weight of $1.$

We then use the following strategy. We arbitrarily initialize $p_i = B/N = 0.1$ for all $i$, i.e., $B$ is uniformly divided across all the $N$ pages. Since these $p_i$ values are arbitrarily chosen, these need not be the optimal crawling rates. Thereafter, we run each of our estimators for (say) 50 iterations. We then use the estimates of $\Delta_i$ at the $50^{th}$ iteration as input to \cite[Algorithm 2]{Azar2018} and obtain the associated possibly sub-optimal crawling rates. Denoting these new rates as $p_i$ again, we now repeat the above procedure. That is, we use $p_i$'s for $50$ iterations to estimate the $\Delta_i$'s and, in turn, use the later to obtain estimates for the new $p_i$'s.

Fig.\ \ref{fig: crawl_rates_comparison} compares the estimated crawling rates obtained using our three estimators with the optimal ones obtained by using the actual $\Delta_i$ values in \cite[Algorithm 2]{Azar2018} for the two kinds of pages. In this experiment, the two SA-based estimators appear to perform better than the LLN estimator, at least initially. Note that $p_5$ denotes the optimal crawling rate for the $5$-th page, while $p_5^{\text{SA}}$ denotes the estimate obtained by using the SA estimator, etc. $P_{10},$ $p_{10}^{SA},$ etc. have similar meanings in relation to the $10$-th page. The parameters we chose for our different estimators are as follows. For the LLN estimator,  we chose  $\alpha_k \equiv 1$. For the SA estimator, we chose $\eta_k = (k + 1)^{-\eta}$ with $\eta = 0.75$. For the SAM estimator, we choose $\eta_k = (k + 1)^{-\eta}$ with $\eta = 1.3$ and $\zeta_k = (\beta_k - \omega \eta_k)/\beta_{k - 1}$ with $\omega = 1,$ and $\beta_k = (k + 1)^{-\beta}$ for $\beta = 0.75$.


\section{Conclusion and Future Work} \label{sec7}

We propose three new online approaches for estimating the rate of change of web pages. We provide theoretical guarantees for their convergence and also provide numerical simulations to compare their performances.  From experiments, one can verify that the proposed estimators perform significantly better than the Naive estimator. Also, they have extremely simple update rules which make them computationally attractive when compared to MLE. We also provide important insights on which estimator one should use in practice.  

The performance of both our estimators currently depend on the choice of $\{\alpha_k\}$, $\{\eta_k\},$ and $\{\zeta_k\}$ respectively. One aspect to analyze in the future would be to ask what would be the ideal choice for these sequences that would help attain the fastest convergence rate. Another interesting research direction to pursue is to combine the online estimation with dynamic optimization. 

\section*{Acknowledgment}
This work is partly supported by ANSWER project PIA FSN2 (P15 9564-266178 \textbackslash  DOS0060094) and DST-Inria  project ``Machine Learning for Network Analytics" IFC/DST-Inria-2016-01/448. Research of Gugan Thoppe is supported by IISc's  start up grants SG/MHRD-19-0054
and SR/MHRD-19-0040. The authors would also like to thank A.~Budhiraja for several useful discussions concerning Theorem~\ref{thm:SAM_Est}.

\bibliographystyle{elsarticle-num}
\bibliography{ref}
\appendix

\section{Convergence of Stochastic Approximation Algorithms}
In this section, we discuss results from literature that provide sufficient conditions for convergence of both one-timescale and two-timescale stochastic approximation algorithms. 

We begin by discussing the convergence of a generic one-timescale stochastic approximation algorithm. This result is obtained by combining \cite[Chapter~2, Corollary~4,]{borkar2009stochastic} and \cite[Chapter~3, Theorem~7]{borkar2009stochastic}.
\begin{theorem}[Convergence of One-timescale Stochastic Approximation \cite{borkar2009stochastic}]
\label{thm:1TS.Conv}
Consider the update rule
\[
    y_{k + 1} = y_k + \eta_k [h(y_k) + M_{k + 1}],
\]
where $\eta_k$ is a positive scalar; $y_k, M_k \in \Real^d;$ and $h: \Real^d \to \Real^d$ is a deterministic function. Suppose the following conditions hold:
\begin{enumerate}[i.)]
    \item $\sum_{k = 0}^\infty \eta_k = \infty$ and $\sum_{k = 0}^\infty \eta_k^2 < \infty.$
    
    \item $\{M_k\}$ is a martingale difference sequence with respect to the increasing family of $\sigma-$fields 
    \[
        \cF_k :=\sigma(y_j, M_j, j \leq k), \quad k \geq 0.
    \]
    That is, $\E[M_{k + 1}|\cF_k] = 0$ a.s., $k \geq 0.$ Further, there is a constant $C \geq 0$ such that $\E[\|M_{k + 1}\|^2|\cF_k] \leq C(1 + \|y_k\|^2) \text{ a.s.}$ for all $k \geq 0.$
        
    \item $h$ is a globally Lipschitz continuous function. Further, the ODE  $\dot{y}(t) = h(y(t))$ has  an unique globally asymptotically stable equilibrium $\yS.$ 
    
    \item There exists a continuous function $h_\infty: \Real^d \to \Real^d$ such that the functions $h_c(x) := h(cx)/c,$ $c \geq 1,$ satisfy $h_c \to h_\infty$ uniformly on compact sets as $c \to \infty.$ Further, the ODE $\dot{y}(t) = h_{\infty}(y(t))$ has the origin as its unique globally asymptotically stable equilibrium.
\end{enumerate}
Then, $y_k \to \yS$ a.s.
\end{theorem}

Often, stochastic approximation algorithms contain an additional perturbation term that is asymptotically negligible. The next result discusses convergence of such algorithms.

\begin{proposition}[Convergence of Perturbed One-timescale Stochastic Approximation]
\label{prop:1TS.Conv.Perturbed}
Consider the update rule
\[
y_{k + 1} = y_k + \eta_k[h(y_k) + \epsilon_k + M_{k + 1}],
\]
where $\epsilon_k$ is an additional perturbation term while the other terms have the same meaning as in Theorem~\ref{thm:1TS.Conv}. Suppose that the four conditions listed in Theorem~\ref{thm:1TS.Conv} hold true. Further, suppose  $\|\epsilon_k\| \leq C \rho_k(1 + \|y_k\|)$ a.s. for $k \geq 0,$ where $C$ is a positive constant and $\{\rho_k\}$ is a sequence of positive scalars such that $\lim_{k \to \infty} \rho_k = 0.$ Then, $y_k \to \yS$ a.s.
\end{proposition}
\begin{proof}
We only give a sketch of the proof since the arguments are more or less similar to the ones used to derive Theorem~\ref{thm:1TS.Conv}. As mentioned before, this latter result follows from \cite[Chapter~2, Corollary~4]{borkar2009stochastic} and \cite[Chapter3, Theorem~7]{borkar2009stochastic}. We now briefly discuss how, even in the presence of the additional perturbation term, these two results continue to hold. 

    \begin{itemize}
    \item \cite[Chapter~2, Corollary~4]{borkar2009stochastic}: This result follows from \cite[Chapter~2, Theorem~2]{borkar2009stochastic} which, in turn, follows from \cite[Chapter~2, Lemma~1]{borkar2009stochastic}. However, as shown in extension $3$ in \cite[pg. 17]{borkar2009stochastic}, this latter result goes through even in the presence of the perturbation term $\{\epsilon_k\}.$ This is because $\epsilon_k$ is asymptotically negligible a.s. More specifically, observe that the sequence $\{y_k\}$ is a.s. bounded under assumption \textbf{(A4)} given on \cite[pg.~17]{borkar2009stochastic}. This implies that $\{\epsilon_k\}$ is a random bounded sequence which is $o(1)$ a.s.; the latter is true since $\rho_k \to 0.$

    \item \cite[Chapter3, Theorem~7]{borkar2009stochastic}: The proof of this result is based on Lemmas 1 to 6 in \cite[Chapter~3]{borkar2009stochastic}. The first three of these lemmas concerns the behaviour of the solution trajectories of the limiting ODE $\dot{y}(t) = h_\infty(y(t)).$ Since the perturbation term does not affect the definition of this limiting ODE in any way whatsoever,  these three results continue to hold as before. Similarly, Lemma~5 in ibid is unchanged since it only concerns the convergence of the sum of  martingale differences $\sum_{k} \eta_k \hat{M}_{k + 1}$  (recall that the stepsize sequence in our update rule is $\eta_k$). With regards to the proof of  Lemma~4 in ibid, observe that our update rule satisfies
    \[
        \hat{y}(t(k + 1)) = \hat{y}(t(k)) + \eta_k(h_{r(n)}(\hat{y}(t(k))) + \hat{\epsilon}_k + \hat{M}_{k + 1}), \quad m(n) \leq k \leq m(n + 1),
    \]
    where $\hat{\epsilon}_k = \epsilon_k/r(n)$ while the other notations  are analogous to the ones defined in \cite[Chapter~3]{borkar2009stochastic}. Because $\|\epsilon_k\| \leq C\rho_k(1 + \|y_k\|),$ $\rho_k \to 0,$ and $r(n) \geq 1,$ it follows that 
    \[
        \|\hat{\epsilon}_k\| \leq C_1(1 + \|\hat{y}(t(k))\|^2)
    \]
    for some positive constant $C_1.$ Note that this is in similar spirit to (3.2.5) in ibid. It is then easy to see that the rest of the proof goes through as before. This  shows that \cite[Chapter~3,Lemma~4]{borkar2009stochastic} continues to be true even in the presence of the the perturbation term. Using exactly the same bound for $\|\hat{\epsilon}_k\|$ obtained above, one can see that the arguments in the proof of Lemma~6 in ibid hold as well. Thus, \cite[Chapter~3, Theorem~7]{borkar2009stochastic} continues to hold, which is exactly what we wanted to establish.
\end{itemize}

The desired result now follows. 
\end{proof}

We next state a result that discusses the convergence of a generic two-timescale stochastic approximation algorithm. The proof of this result is based on \cite[Chapter 6, Theorem~2]{borkar2009stochastic} and \cite[Theorem~10]{lakshminarayanan2017stability}.

\begin{theorem}[Convergence of Two-timescale Stochastic Approximation \cite{borkar2009stochastic, lakshminarayanan2017stability}]
\label{thm:2TS.Conv}
Consider the update rules
\begin{align*}
    u_{k + 1} = {} & u_k + \gamma_k [h(u_k, z_k) + M^{(1)}_{k + 1}], \\
    z_{k + 1} = {} & z_k + \beta_k[g(u_k, z_k) + M^{(2)}_{k + 1}],
\end{align*}
where $\gamma_k$ and $\beta_k$ are  positive scalars; $u_k, z_k, M^{(1)}_k, M^{(2)}_k \in \Real^d;$ and $h, g: \Real^{2d} \to \Real^d$ are two deterministic functions. Suppose the following conditions hold:
\begin{enumerate}[i.)]

    \item $\sum_{k \geq 0} \gamma_k  = \sum_{k \geq 0} \beta_k = \infty,$  $\sum_{k \geq 0} \left(\gamma_k^2 + \beta_k^2 \right) < \infty,$ and $\lim_{k \to \infty} \dfrac{\beta_k}{\gamma_k} = 0.$

    \item $\{M_k^{(1)}\}$ and $\{M_k^{(2)}\}$ are martingale difference sequences with respect to the increasing $\sigma-$fields
    \[
        \cF_k := \sigma(u_j, z_j, M_j^{(1)}, M_j^{(2)}, j \leq k), \quad k \geq 0.
    \]
    Further, there exists a constant $C \geq 0$ such that $\E[\|M_{k + 1}^{(i)}\|^2|\cF_k] \leq C(1 + \|u_k\|^2 + \|z_k\|^2)$ for $i = 1, 2$ and $k \geq 0.$
    
    \item $h$ and $g$ are globally Lipschitz continuous functions. For each fixed $z,$ the ODE $\dot{u}(t) = h(u(t), z)$ has a unique globally asymptotically stable equilibrium $\phi(z),$ where $\phi: \Real^d \to \Real^d$ is Lipschitz continuous. Further, the ODE $\dot{z}(t) = g(\phi(z(t)), z(t))$ has an unique globally asymptotically stable equilibrium $\zS.$ 
    
    \item The functions $h_c(u, z) := h(cu, cz)/c,$ $c \geq 1,$ satisfy $h_c \to h_\infty$ as $c \to \infty,$ uniformly on compacts for $h_\infty.$ Also, for each fixed $z \in \Real^d,$ the limiting ODE $\dot{u}(t) = h_\infty(u(t), z)$ has a unique globally asymptotically stable equilibrium $\phi_\infty(z),$ where $\phi_\infty: \Real^d \to \Real^d$ is a Lipschitz map. Further, $\phi_\infty(0) = 0.$ Separately, the functions $g_c(z):= g(c \phi_\infty(z), cz)/c,$ $c \geq 1,$ satisfy $g_c \to g_\infty$ as $c \to \infty,$ uniformly on compacts for some $g_\infty.$ Also, the limiting ODE $\dot{z}(t) = g_\infty(z(t))$ has the origin as its unique globally asymptotically stable equilibrium. 
\end{enumerate}
Then, $(u_k, z_k) \to (\phi(\zS), \zS)$ a.s.
\end{theorem}

The last and final result of this section concerns the convergence of two-timescale stochastic approximation with perturbation terms that are asymptotically negligible. 

\begin{proposition}[Convergence of Perturbed Two-timescale Stochastic Approximation]
\label{prop:2TS.Conv.Perturbed}
Consider the update rules
\begin{align*}
u_{k + 1} = {} & u_k + \gamma_k[h(u_k, z_k) + \epsilon^{(1)}_k + M^{(1)}_{k + 1}] \\
z_{k + 1} = {} & z_k + \beta_k[g(u_k, z_k) + \epsilon^{(2)}_k + M^{(2)}_{k + 1}],
\end{align*}
where $\epsilon^{(1)}_k, \epsilon^{(2)}_k$ are additional perturbation terms while the other terms have the same meaning as in Theorem~\ref{thm:2TS.Conv}. Suppose that the four conditions listed in Theorem~\ref{thm:2TS.Conv} hold true. Further, suppose  $\|\epsilon^{(i)}_k\| \leq C \rho^{(i)}_k(1 + \|u_k\| + \|z_k\|)$ a.s. for $k \geq 0$ and $i = 1,2,$ where $C$ is a positive constant and $\{\rho^{(i)}_k\},$ $i = 1, 2,$ are sequences of positive scalars such that $\lim_{k \to \infty} \rho^{(i)}_k = 0.$ Then, $(u_k, z_k) \to (\phi(\zS), \zS)$ a.s.
\end{proposition}
\begin{proof}
As stated before, this result follows from  \cite[Chapter 6, Theorem~2]{borkar2009stochastic} and \cite[Theorem~10]{lakshminarayanan2017stability}. We now briefly discuss how these results continue to hold even in the presence of the perturbation terms $\epsilon^{(1)}_k$ and $\epsilon^{(2)}_k.$

\begin{itemize}
    \item \cite[Chapter 6, Theorem~2]{borkar2009stochastic}: This result, as well as $\cite[Chapter~6, Lemma~1]{borkar2009stochastic}$ on which it relies, are essentially proved by defining suitable one-timescale stochastic approximation algorithms and then using  convergence results concerning the latter. In our situation, both these will have additional perturbation terms that are asymptotically negligible. Consequently, by arguing as in the third extension given in \cite[pg.~27]{borkar2009stochastic}, it can be shown that the   asymptotic behaviour of these two algorithms remains unchanged even in the perturbed setup. Therefore, it follows that the conclusions of \cite[Chapter~6, Theorem~2]{borkar2009stochastic} continue to hold as before. 
    
    \item \cite[Theorem~10]{lakshminarayanan2017stability}: This result is based on Lemmas~2 to 7 and Lemma 9 as well as Theorems~6 and 7 in ibid. Lemmas~2 to 5 in ibid concern the limiting ODEs described in condition iv.) of Theorem~\ref{thm:2TS.Conv} above. The definitions of these ODEs do not depend on the presence or absence of the perturbation terms. Therefore, the aforementioned four lemmas continue to hold as before. On the other hand, Lemmas~6 and 9 in ibid rely on the results in Chapter~3 and Chapter~6 of  \cite{borkar2009stochastic}. As argued before, these results continue to hold even in the presence of perturbation terms and, consequently, so do Lemmas~6 and 9 in ibid. Finally, Theorems~8 and 10 in ibid build upon these seven Lemmas. Therefore, they hold as well in the perturbed setup. 
\end{itemize}
The desired result now follows. 
\end{proof}
\end{document}